\documentclass[a4paper,10pt]{scrartcl}

\RequirePackage{xr}
\RequirePackage[OT1]{fontenc}
\RequirePackage{amsmath,amssymb,dsfont,bm}
\RequirePackage{amsthm}
\usepackage[round]{natbib}
\RequirePackage{algorithm}
\RequirePackage{algpseudocode}
\RequirePackage{graphicx} 
\RequirePackage[section]{placeins}
\RequirePackage{color}
\RequirePackage{chngcntr}
\RequirePackage{enumerate}
\usepackage[shortlabels]{enumitem}
\RequirePackage{longtable}
\usepackage{multirow}
\usepackage[english]{babel} 
\usepackage{xcolor}
\usepackage{bbm}
\usepackage{mathtools}
\usepackage{tikz}
\usetikzlibrary{positioning}
\usetikzlibrary{calc}
\usepackage{xfrac}
\usepackage{verbatim}
\usepackage{url}
\usepackage{mathtools}
\usepackage{lmodern}
\usepackage{breqn}
\usepackage{colortbl}
\usepackage[title]{appendix}
\usepackage{setspace}

\algnewcommand\algorithmicinput{\textbf{Input:}}
\algnewcommand\INPUT{\item[\algorithmicinput]}

\algnewcommand\algorithmicret{\textbf{Return:}}
\algnewcommand\RETURN{\item[\algorithmicret]}

\counterwithin{table}{section}
\counterwithin{figure}{section}

\graphicspath{ {figures/} }

\newcommand{\abs}[1]{\ensuremath{\left\vert#1\right\vert}}
\newcommand\ZuWeis{\mathrel{\mathop:\!\!=}} 
\newcommand\WeisZu{\mathrel{=\!\!\mathop:}} 
\newcommand{\R}{\ensuremath{\mathds{R}}}  
\newcommand{\Pp}{\ensuremath{\textbf{P}}}  
\newcommand{\argmax}{\operatorname{argmax}} 
\newcommand{\argmin}{\operatorname{argmin}} 
\newcommand{\cH}{\mathcal{H}}
\newcommand{\cN}{\mathcal{N}}
\newcommand{\indE}{\mathds{1}} 
\newcommand*{\e}{\textup{e}}
\newcommand{\cpset}{\Sigma}
\newcommand{\Sna}{\hat{S}_{n,\alpha}}
\newcommand{\qna}{q_n(\alpha)}
\newcommand{\pen}{P_{\ell,n}}

\newcommand{\mdl}{\Lambda}

\newtheorem{example}{Example}[section]
\newtheorem{theorem}{Theorem}[section]

\newtheorem{lemma}{Lemma}[section]

\newtheorem{definition}{Definition}[section]
\newtheorem{remark}{Remark}[section]

\begin{document}

\begin{center}
\begin{minipage}{.8\textwidth}
\centering 
\LARGE Multiscale quantile segmentation\\[0.5cm]

\normalsize
\textsc{Laura Jula Vanegas$^\ast$, Merle Behr$^\ddagger$ and Axel Munk$^{\ast,\dag}$}\\[0.1cm]
Institute for Mathematical Stochastics$^\ast$,\\
University of G\"ottingen,\\
Max Planck Institute for Biophysical Chemistry$^\dag$,\\
G\"ottingen, Germany,\\
and\\
Department of Statistics$^\ddagger$,\\
University of California at Berkeley,
Berkeley, USA\\
Email: \verb+{ljulava,munk}@math.uni-goettingen.de, behr@berkeley.edu+

\end{minipage}
\end{center}

\begin{abstract}
We introduce a new methodology for analyzing serial data by quantile regression assuming that the underlying quantile function consists of constant segments.
The procedure does not rely on any distributional assumption besides serial independence. 
It is based on a multiscale statistic, which allows to control the (finite sample) probability for selecting the correct number of segments $S$ at a given error level, which serves as a tuning parameter. 
For a proper choice of this parameter, this tends exponentially fast to the true $S$, as sample size increases. 
We further show that the location and size of segments are estimated at minimax optimal rate (compared to a Gaussian setting) up to a log-factor.  
Thereby, our approach leads to (asymptotically) uniform confidence bands for the entire quantile regression function in a fully nonparametric setup.
The procedure is efficiently implemented using dynamic programming techniques with double heap structures, and software is provided.
Simulations and data examples from genetic sequencing and ion channel recordings confirm the robustness of the proposed procedure, which at the same time reliably detects changes in quantiles from arbitrary distributions with precise statistical guarantees.
\end{abstract}



\section{Introduction}\label{sec:intro}
The analysis of serial data with presumably abrupt underlying distributional changes, for example, in its mean, median, or variance, is a long-standing issue and relevant to a magnitude of applications, e.g. to econometrics and empirical finance \citep{preuss2015, shen2016, russell2019}, evolutionary and cancer genetics \citep{liu2013, zhang2007, jonas2016} or neuroscience \citep{cribben2017}, to mention a few (see Section \ref{subsec:rel} for a more comprehensive discussion).
%
%
%
The present work 
proposes a new methodology for this task, denoted as \textit{Multiscale Quantile Segmentation (MQS)} which on the one hand, is extremely robust as it does not rely on any distributional assumption (apart from independence) and on the other hand, still has high detection power with (even non-asymptotic) statistical guarantees.
More precisely, our methodology is based on quantile segments, i.e., quantile regression functions which are modeled as right-continuous piecewise constant functions $\vartheta$ with finitely but arbitrary many (unknown) segments $S$.
We stress that even in such situations where the quantiles are not piecewise constant, this may serve as a reasonable proxy to cartoonize the quantile function in a simple but meaningful way, resulting in a ``quantilogram'' similar in spirit to Tukey's regressogram \citep{tukey1961}.
To fix our setting, we assume that the underlying regressor (e.g., time) is in the interval $[0,1)$ and is sampled equidistantly at $n$ sampling points $x_i\ZuWeis (i-1)/n$. 
As our main results are nonasymptotic, extensions to general (ordered) sampling domains and non-equidistant sampling points are immediate. 
All such segment functions $\vartheta:[0,1)\rightarrow\R$ are then comprised in the space
\begin{equation}\label{eq:Sset} 
\cpset=\left\{\vartheta =\sum_{s=1}^S\theta_s \indE_{[\tau_{s-1},\tau_{s})} : \, \theta_s\neq\theta_{s+1},\right. \\
\left. 0=\tau_0<\tau_1<...<\tau_{S}=1,\, S<\infty\right\}.
\end{equation}
Hence, our quantile function consist of $S$ (unknown) distinct segments with unknown segment values $\theta_s\in\R$ and segment lengths $\tau_s-\tau_{s-1}$ (see Figure \ref{fig:sex} for illustration). 
%
\medskip
\begin{mod}(\textsc{Quantile Segment Regression model})\label{mod:model1}
Let $\beta\in(0,1)$ and $\vartheta_\beta \in \cpset$ a segment function. In the QSR-model one observes $n$ independent random variables $Z_i$ at equidistant sampling points $x_{i,n}=x_i=(i-1)/n$ for $i=1,\ldots n$ such that its $\beta$-quantiles are given as
\begin{align}\label{eq:model1}
\vartheta_{\beta}(x_i) = \inf \{ \theta :\; \Pp(Z_i \leq \theta) = \beta \} \quad\mbox{ for  }i=1,\ldots,n.
\end{align}
\end{mod}
%
Note that for a particular observed signal $Z_i$, $i = 1, \ldots, n$, different quantiles $\beta \neq \beta^\prime$ will, in general, correspond to different segment functions $\vartheta_\beta \neq \vartheta_{\beta^\prime} \in \Sigma$ in (\ref{eq:Sset}).
Here we focus on finding the corresponding segment function for a particular, but arbitrary, given quantile $\beta$, e.g., the median ($\beta = 0.5$). 
Screening for several $\beta$-values then will allow to specify in which part of the distribution significant changes occur.
This is in contrast to non-parametric distributional segmentation (see Section \ref{subsec:rel} for references), where one would aim to find changes in the distribution per se, i.e. a segment change is detected when (at least) some of the quantile functions changes.
Therefore, we do not consider $\beta$ as a tuning parameter, but rather as a user-specific input, that depends on the particular task of interest.

\begin{example}\label{ex:intro}
Figure \ref{fig:sex} shows a data example from an ion channel recording experiment performed by the Steinem lab (Insitute of Organic and Biomolecular Chemistry, University of G\"ottingen) (see more details in Section \ref{subsec:ion}).
Using the patch-clamp technique \citep{sakmann1995b}, one can measure the current flow of ions transported through an individual channel in a cell's membrane. 
Figure \ref{fig:sex} shows data from the bacterial porin PorB, an outer membrane porin of \textit{Neisseria meningitidis} -- a pathogenic bacteria well known for being the agent of epidemic meningitis \citep{virji2009}.
Over time, the ion channel changes its gating behavior by closing and reopening its pore.
This leads to a piecewise constant current flow structure, shown as black dots in Figure \ref{fig:sex}.
The opening and closing behavior are of immediate physiological interest, and we can consider this within the QSR-model, where in this data example $n = 2,801$.
By targeting different quantiles, one can recover different properties of the channel.
For example, the median $\beta = 0.5$ provides a measure of average current flow and can thus quantify the channel's overall configuration.
The second row in Figure \ref{fig:sex} shows in red our MQS segmentation for $\beta = 0.5$ and $\alpha = 0.01$, with fully automatic estimated number of segments $\hat{S}=10$ (recall (\ref{eq:Sset})). 

In addition, upper and lower tail quantiles provide more information on the variability of the channel behavior. The third and fourth rows in Figure \ref{fig:sex} show in light red our MQS segmentation for $\beta = 0.25$ and $\beta = 0.75$, respectively.
Moreover, the first row in Figure \ref{fig:sex} shows together the MQS' estimates for the 0.25, 0.5, and 0.75-quantiles. 
We refer to this visualization of the three quantiles as the \textit{multiscale segment boxplot} (MSB).
This shows that the variance on the lower state is lower than on the higher state, a common characteristic of many channels, known as open channel noise, see \cite{sakmann1995b}, which is well known to aggravate data analysis. By targeting different quantiles, MQS provides a direct approach for its quantification. 
Besides estimates for the respective quantile segmentation, our procedure also comes with a precise uncertainty quantification.
This is illustrated with confidence intervals for the segment locations (blue lines) and confidence bands for the segment function (grey areas).
We stress that these confidence statements do not rely on any distributional assumptions for the data generating process (besides independence), which is in contrast to other available approaches for this purpose, see Section \ref{subsec:rel}.
This makes the MQS procedure particularly useful in practice, where parametric distributional assumptions are often problematic.
For example, although a Gaussian assumption is widely used for ion channel data, see e.g., \cite{pein2017a}, performing a Shapiro-Wilk's \citep{shapiro1965} test on the first segment of the present ion channel data rejects with p-value = 0.0016.
\begin{figure}[th!]
\centering
\includegraphics[width=0.95\textwidth]{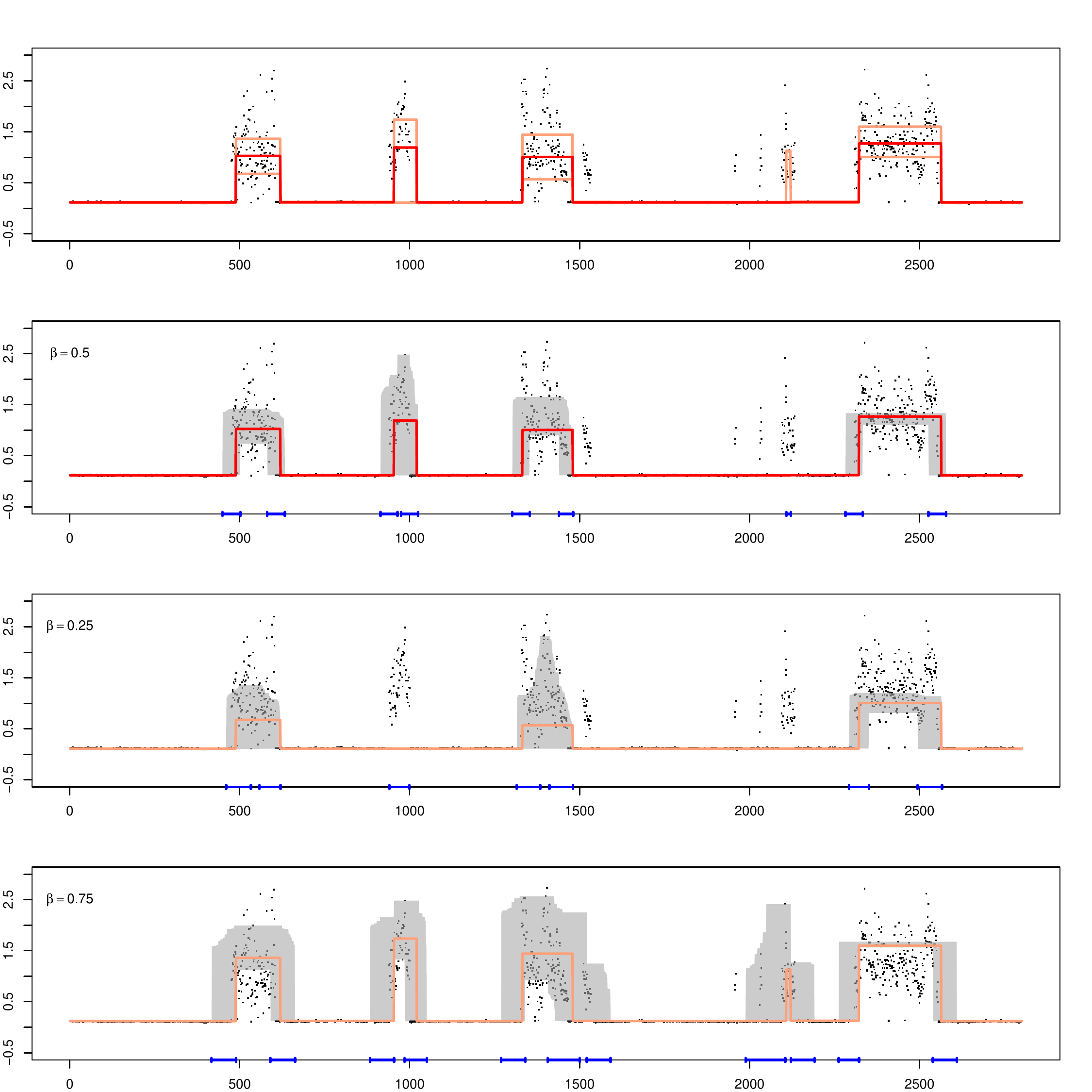}
\caption{\small First row: Observations $Z_1,\ldots,Z_n$, $n=2801$, from ion channel data recorded at the Steinem lab (Insitute of Organic and Biomolecular Chemistry, University of G\"ottingen),
and the multiscale segment boxplot (MSB), with estimates (MQSE) for the median (red line), the 0.25- and 0.75-quantiles (light red lines), at nominal level $\alpha = 0.01$, see (\ref{eq:1statg}). 
Subsequent rows: MQSE (red and light red), together with  $99\%$ simultaneous confidence bands (gray area) and simultaneous confidence intervals for the change-point locations (blue intervals), for $\beta=0.5,0.75,0.25$, respectively. 
}
\label{fig:sex}
\end{figure}
\end{example}
As illustrated in Example \ref{ex:intro}, the aim of this work is to provide a statistical methodology for multiscale quantile segmentation (MQS) in the general QSR-model. 
In particular, we stress that MQS is robust to arbitrary distributional changes that keep the respective quantile under consideration unchanged, including changes of variance, as demonstrated in Example \ref{ex:intro}.
See also Appendix \ref{sec:example2}, which further illustrates this robustness with some synthetic simulated data example.
For each $\beta\in(0,1)$, MQS provides estimates and confidence statements for:
\begin{enumerate}
\item the number $S$ of segments $S(\vartheta_\beta)$,
\item the segment locations $\tau_1,\ldots,\tau_{S-1} \in [0,1)$,
\item and, based on these, the segment values, i.e., the quantiles $\theta_1,\ldots,\theta_S \in \R$.
\end{enumerate}
Our approach is based on a simple transformation: 
Given a candidate segment function $\vartheta \in \cpset$ (which will depend on data $Z_i$), we consider the transformed binary (pseudo) data
\begin{equation}\label{eq:transformtion}
W_i=W_i(Z_i,\vartheta(x_i)) \ZuWeis \begin{cases}
       0& \mbox{  if  }\vartheta(x_i)-Z_i<0,\\
			 1& \mbox{  if  }\vartheta(x_i)-Z_i\geq 0.
			\end{cases}
\end{equation} 
Note that if and only if the candidate function $\vartheta$ equals the true underlying regression function $\vartheta_\beta$ in the QSR-model, then $W_1,\ldots,W_n$ are i.i.d.\ Bernoullis with success probability $\beta$.
Based on this, our methodology ``tests" in a multiscale fashion any possible candidate function in $\cpset$ to be valid: it selects a segment quantile function which does not contradict the i.i.d.\ Bernoulli assumption and, among those, has the smallest number of segments.

To this end, the unknown number and locations of segments of $\vartheta_\beta$ will be detected by a certain multiscale statistic $T_n(Z,\vartheta) = T_n(W(Z,\vartheta), \vartheta)$ (see (\ref{eq:mts})), which combines on intervals $[x_i,x_j]$ where $\vartheta|_{[x_i,x_j]}$ is constant, with $1 \leq i \leq j \leq n$, the corresponding log-likelihood ratio tests for the hypothesis testing problems
\begin{equation}\label{eq:introLocHyp}
H_{ij}: W_i, \ldots,W_j \overset{i.i.d.}{\sim} B(\beta) \quad \text{vs.} \quad K_{ij}: W_i, \ldots,W_j \overset{i.i.d.}{\sim} B(\beta^\prime)\text{ with }\beta^\prime \neq \beta,
\end{equation}
where $B(\beta)$ denotes a Bernoulli distribution with success probability $\beta$ (see (\ref{eq:mts}) for the precise definition of $T_n$). 
In a first step MQS determines the number of segments from the data given such $\beta\in(0,1)$. 
For a given threshold $q \in \R$ only depending on $\beta$, $n$, and the nominal level $\alpha$ (to be specified later), the estimated number of segments $\hat{S}$ is then the solution of a (non-convex) optimization problem with convex constrains given by the multiscale statistic $T_n$, namely,
\begin{align}\label{eq:hatK}
\hat{S}=\hat{S}(q) \ZuWeis \inf_{\vartheta \in \cpset } \# J(\vartheta)\quad \text{   s.t   } \quad T_n(W(Z,\vartheta), \vartheta) \leq q,
\end{align}
where $J(\vartheta)$ denotes the set of segments and $\# J(\vartheta)$ the number of segments of $\vartheta \in \cpset$.
That is, we choose the smallest number of segments such that the multiscale test $T_n(W(Z,\vartheta), \vartheta) \leq q$ still accepts. As, in a certain sense, this is a model selection step, we call $\hat{S}$ the MQS selector of S.
In a second step, we then constrain all candidate segment quantile functions to those with $\hat{S}$ segments, that is, to the set
\begin{align}\label{eq:H}
\cH(q) \ZuWeis \{\vartheta \in \cpset : \#J(\vartheta) = \hat{S}(q) \text{ and } T_n(W(Z,\vartheta), \vartheta) \leq q\}.
\end{align}
The MQS estimate (MQSE) $\hat{\vartheta}$ is then a particular simple segment function in $\cH(q)$ based on the Wald-Wolfowitz runs tests (see Section \ref{sec:theory} for details and Figure \ref{fig:sex} for illustration).
\subsection{Main results: Theory}

For the MQSE $\hat{\vartheta}$, with $\vartheta_\beta$ being the true underlying regression function in the QSR-model, we immediately get from (\ref{eq:hatK}) that 
$\mathbb{P}\left(\hat{S}(q)>S\right)\leq \mathbb{P}\left(T_n\left(Z,\vartheta_\beta\right)\leq q\right).$
In Section \ref{subsec:theoMS} we show that $T_n(Z,\vartheta_\beta)$ can be bounded in distribution with a random variable $M_n = M_{n,\beta}$ that only depends on $\beta$ and $n$ and hence does not depend on $\vartheta_\beta$ or any other characteristics of the underlying distribution of the observations $Z_i$, namely
\begin{equation}
    M_n\ZuWeis \max_{1\leq i\leq j\leq n}\sqrt{2T_i^j(X,\beta)}-\pen,
\label{eq:mn}
\end{equation}
where $X=(X_1,\ldots,X_n)$ are i.i.d. Bernoulli distributed with mean $\beta$, and a penalization $\pen$ depending only on $n$ and $\ell = j-i+1$, see (\ref{eq:mts}). Therefore, for given $n$ and $\beta$ the finite sample-quantiles of $M_n$ can be computed by Monte-Carlo simulations in a universal manner. 
Hence, choosing $q_n=q_{n}(\alpha)$ as the $(1-\alpha)$-quantile of $M_n$ gives us control for the overestimation error of the number of segments included in our final estimator by a desired level $\alpha$, namely
\begin{equation}
\mathbb{P}\left(\Sna>S\right)\leq\alpha,\label{eq:1statg}
\end{equation}
where $\Sna\ZuWeis\hat{S}\left(q_n(\alpha)\right)$. 
In Theorem \ref{theo:overestimation} we present a refinement of this fact.
Note that $q_{n}(\alpha)$ depends on $\beta$, however, asymptotically this dependency vanishes, see Remark \ref{rem:limitDistM}.
\begin{theorem}[Overestimation error]\label{theo:overestimation}
Consider the QSR-model. For $q = q_n(\alpha)$ as in (\ref{eq:hatK}) and (\ref{eq:quantile}) the MQS-selector $\hat{S}_{n,\alpha}:=\hat{S}(q)$ satisfies uniformly over all segment functions $\vartheta\in\cpset$
\begin{equation}\label{eq:introOverest}
\mathbb{P}\left(\Sna>S+s\right)\leq \alpha^{\lfloor s/2\rfloor+1},
\end{equation}
where $\lfloor \cdot\rfloor := \max\{m\in\mathbb{Z}:m\leq x\}$, for all $x\in\mathbb{R}$, i.e. the floor function.
\end{theorem}
This means, in addition to controlling the overall error to include at least one extra segment ($s=0$), the error of overestimating the number of segments by more than $s$ decays exponentially fast. This reveals that segments detected by MQS are, with very high probability, indeed present in the signal. Therefore, (\ref{eq:introOverest}) guarantees that $\alpha$ controls the false positives in a strong family-wise error sense.
The overestimation bound in Theorem \ref{theo:overestimation} is complemented by an explicit bound for the underestimation error, i.e. to miss a segment, see Theorem \ref{cor:easyunder}.
Together, this allows for precise fine-tuning of the error of a wrong number of detected segments via the choice of the error level $\alpha$.
Clearly, such an underestimation bound for $\mathbb{P}(\Sna<S)$ has to depend on some characteristics of the function $\vartheta_\beta$, specifically on the length and height of the jumps, as no method can detect arbitrary small changes for a fixed number of data.
Moreover, note that even a large jump of $\vartheta_\beta$ is not identifiable from the QSR-model, if it does not induce a sufficiently large jump in the respective distribution functions, as the following example shows.
\begin{example}\label{ex:into2}
For sufficiently large $L >0$ and small $\epsilon > 0$ consider random variables $X, Y$ such that $\Pp(X = L) = 0.5$, $\Pp(X = -L) = 0.5 - \epsilon$, $\Pp(X = -L + \epsilon) = \epsilon$ and $\Pp(Y = - L) = 0.5-\epsilon$, $\Pp(Y = L) = 0.5$, $\Pp(Y = L- \epsilon) = \epsilon$. Then, if half of the sample comes from $X$, i.e., $Z_1,\ldots,Z_{n/2} \;\sim \; X$ and the other half from $Y$, i.e., $Z_{n/2 + 1}, \ldots, Z_n \;\sim\; Y$ as in the QSR-model, the underlying median regression function $\vartheta_{0.5}$ has a large jump of size $2 L - 2 \epsilon$ at $\tau = 0.5$. However, for sufficiently small $\epsilon$ (depending on $n$), because the probability of sampling $-L+\epsilon$ or $L-\epsilon$ is very small, this jump is not detectable from the observations $Z_1,\ldots,Z_n$.
See Figure \ref{fig:detex} for illustration.
The following definition specifies situations where a quantile segment is detectable.
\end{example}
\begin{figure}[h!]
\centering
\includegraphics[width=\textwidth]{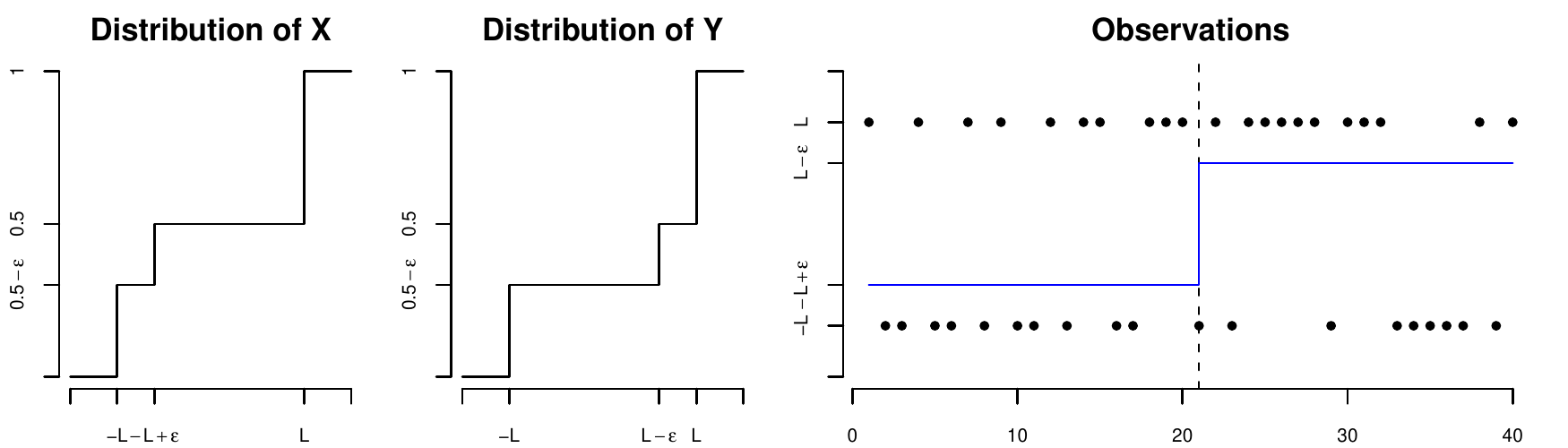}
\caption{\small Left: Cumulative distribution function of $X$ and $Y$ as in Example \ref{ex:into2} of the main text. Right: Independent observations $Z_1,\ldots,Z_{20}\sim X$ and $Z_{21},\ldots,Z_{40}\sim Y$ together with a median regression function (blue line) as in the QSR-model. Here $n=40$, $L=1$, and $\epsilon = 0.01$.
}
\label{fig:detex}
\end{figure}
\begin{definition}\label{def:quantijump}
For a distribution function $F$ and $\beta\in(0,1)$, let $\theta_\beta \ZuWeis \inf\{\theta  :  F(\theta) \geq \beta\}$ be the $\beta$-quantile. Then the \textit{quantile jump function} $\xi_{F,\beta} \colon \mathbb{R}\rightarrow[0,1]$ is defined as
\[\xi_{F,\beta}(\delta)=\abs{F(\theta_\beta+\delta)-\beta}.\]
\end{definition}
The quantile jump function quantifies how much a \textit{quantile jump} in the QSR-model influences data around a jump. 
The bound for the underestimation error for the QSR-Model naturally depends on the minimal length of constant segments and on the minimal quantile jump for the distributions $F_1,\ldots,F_n$ of $Z_1,\ldots, Z_n$, that is,
\begin{equation}\label{eq:LambdaXi}
\Xi \ZuWeis \min_{s=1,\dots,S}  \min\left\{\xi_{F_{\min},\beta}(\theta_{s}-\theta_{s-1}), \xi_{F_{\max},\beta}(\theta_{s-1}-\theta_{s})\right\}\;\mbox{and}\;
\mdl \ZuWeis \min_{s=1,\dots,S}\abs{\tau_{s} - \tau_{s-1}}.
\end{equation}
where $F_{\text{min}}\ZuWeis \min\{F_1,\ldots,F_n\}$ and $
F_{\text{max}}\ZuWeis \max\{F_1,\ldots,F_n\}$ (pointwise).
In Theorem \ref{theo:underestimation} we show a general exponential bound for underestimating the number of segments by the MQSE, which provides the following result.
\begin{theorem}\label{cor:easyunder}
Consider the QSR-model, $\Sna$ as in (\ref{eq:1statg}), and $\Lambda,\Xi$ as in (\ref{eq:LambdaXi}). Then, for $\qna>0$ the $1-\alpha$ quantile of $M_n$ in (\ref{eq:mn}), it holds that
\begin{equation}\mathbb{P}\left(\Sna<S\right)\leq 4(S-1)\textup{e}^{-n\Lambda\Xi^2}\left[\textup{e}^{2\sqrt{n\Lambda}\Xi\left(\qna/\sqrt{2}+\sqrt{\log(\sfrac{2\textup{e}}{\Lambda})}\right)}+1\right].\label{eq:2statg}
\end{equation}
\end{theorem}
Note that whenever $\beta \to 0$ or $\beta \to 1$ then $\Xi \to 0$ and hence, the bound on the r.h.s. in (\ref{eq:2statg}) becomes trivial. This reflects the fact that for very high or very small quantiles it is arbitrarily difficult to capture changes from finitely many samples.
Combining (\ref{eq:1statg}) and (\ref{eq:2statg}), we obtain an explicit bound for estimating the number of segments correctly, depending on $\alpha(q_n) = \mathbb{P}(M_n>q_n)$, $n$, $\Lambda$ and $\Xi$ in (\ref{eq:LambdaXi}), namely,
\begin{equation}
\mathbb{P}\left(\Sna=S\right)\geq 1-\alpha(q_n)- 4(S-1)\textup{e}^{-n\Lambda\Xi^2}\left[\textup{e}^{2\sqrt{n\Lambda}\Xi\left(q_n/\sqrt{2}+\sqrt{\log(\sfrac{2\textup{e}}{\Lambda})}\right)}+1\right],\label{eq:conbound}
\end{equation}
uniformly over all possible segment functions in $\cpset$ satisfying (\ref{eq:LambdaXi}).
As $\alpha(q_n)$ converges to zero as $q_n \to \infty$ (see Section \ref{sec:theory}), the MQS selector is consistent, that is $\mathbb{P}\left(\Sna=S\right)\rightarrow 1$ as $n\rightarrow\infty$,
whenever the threshold parameters are chosen such that  $q_n \to \infty$ and $q_n/\sqrt{n}\rightarrow 0$. Moreover, $q_n$ can be chosen such that the r.h.s of (\ref{eq:conbound}) is maximized leading to exponentially fast selection consistency.
In Theorem \ref{theo:vanishingrates} we refine these consistency results to the situation where $\Xi = \Xi_n$ and $\Lambda = \Lambda_n$ vanish as $n\to\infty$.
We give a sharp condition on $\Xi_n,\Lambda_n,q_n$ under which MQS consistently estimates the number of segments $S$ and show that (up to constants) these conditions cannot be improved, in general. 

With the choice of $q_n=q_n(\alpha)$ as in (\ref{eq:1statg}), the multiscale approach of MQS directly yields (asymptotically) honest simultaneous confidence bands (i.e. uniformly over all possible segment functions in $\cpset$ with minimal scale $\Lambda$ and minimal distribution jump $\Xi$) for the underlying regression function $\vartheta_\beta$ and simultaneous confidence intervals for the change-points $\tau$ via the set $\cH(q_n)$ in (\ref{eq:H}), see last three plots of Figure \ref{fig:sex} and Theorem \ref{theo:confidenceset} for further details.
Besides model selection consistency and confidence statements for all quantities, the MQS procedure also attains minimax optimal estimation rates (up to a log-factor) for the change-point\ locations. In Theorem \ref{theo:cprate} we show that whenever $q_n = o(\sqrt{\log(n)})$ and $\Lambda_n^{-1} = o(n/\log(n))$ then for any $\Xi_0>0$
\begin{equation}\label{eq:intoEstRate}
\sup_{\substack{\vartheta\in\cpset \\ \Lambda>\Lambda_n, \Xi>\Xi_0}}\mathbb{P}\left(\max_{\tau\in J(\vartheta)}\min_{\hat{\tau}\in J(\hat{\vartheta}(q))}|\hat{\tau}-\tau|> \frac{\log(n)}{n} \frac{1}{\Xi^2} \right) \to 0,
\end{equation}
where the minimax rate is lower bounded by the sampling rate $1/n$ and hence, the rate in (\ref{eq:intoEstRate}) is optimal (up to the log-factor).
\subsection{Implementation}
It has been exploited for a long time that global optimization procedures to detect segment changes in parametric models (mainly Gaussian), can be computed efficiently using dynamic programming, see e.g., \citep{bellman1954,friedrich2008, boysen2009, killick2012, davies2012,frick2014, zou2014, pein2017a, haynes2017, celisse2018, wang2018}. 
However, the exact computation of the nonparametric quantile segments of the MQSE in (\ref{eq:hatK}) leads to an additional computational burden compared to the case where the underlying jump signal corresponds to a parameter of a specific data distribution. 
Whereas the parametric case typically involves updating a running mean, which can be done in $\mathcal{O}(1)$ time, the quantile case involves updating the empirical quantile, which depends on the ordering of the data and hence, is computationally more involved.
We incorporate an efficient running quantile computation as described in \citep{astola1989}, which is based on double heap structures. 
Thereby, the worst case computation time of MQSE only increases by a log-factor compared to the parametric case, as e.g., discussed in \citep{frick2014}, being of the order $\mathcal{O}(n^2\log n)$.
However, in many situations, depending on the underlying signal, the actual complexity will be almost linear.
We give more details on the implementation of MQSE in Section \ref{sec:implementation}.
An R package \texttt{mqs} is available at \url{https://github.com/ljvanegas/mqs}.

\subsection{Simulation results and data examples}
In Section \ref{sec:simulations} we explore MQS in a comprehensive simulation study, including a comparison with several other state of the art segmentation methods, which have been designed to be robust, namely R-FPOP \citep{fearnhead2017}, WBS \citep{fryzlewicz2014}, HSMUCE \citep{pein2017a}, QS \citep{eilers2005}, NWBS \citep{padilla2019}, and NOT \citep{baranowski2019}. 
NOT has two contrast functions for different types of robustness, namely "pcwsConstMeanHT" (abbreviated HT) and "pcwsConstMeanVar" (abbreviated VAR) for heavy-tailed noise and changes in variance, respectively.  
In order to compare the detection power of MQSE in a benchmark scenario, we also compare with SMUCE \citep{frick2014} which is tailored to i.i.d. normal error.
A major finding is that while other methods only work well in some specific cases, MQS reliably detects changes in the quantiles for arbitrary distributions. 

Figure \ref{fig:rob} illustrates a benchmark scenario which shows data from different distributional regimes together with the MQS, SMUCE, R-FPOP, WBS, HSMUCE, QS, NOT, and NWBS estimates (red lines) for the median. 
In order to show robustness of the MQSE to changes in distribution we choose data from 3 different distributional regimes. 
The first 350 observations are drawn from a normal distribution with variance 2, the next 1190 from a t distribution with 1 d.f. and variance 0.1 and the last 945 from a $\chi^2$ distribution with 1 d.f. and variance 0.1. 
MQS turns out to have good detection power (see second regime), to be robust to heavy tails and change in variance (see first and second regime), and more generally to arbitrary distributional variations and skewness (see third regime).
In contrast, SMUCE (specifically designed for Gaussian, homogeneous noise) is very sensitive to changes in variance and heavy-tailed distributions, adding a lot of artificial changes.
WBS and NOT(VAR) are more robust to high variances, but also add artificial changes in the presence of outliers.
R-FPOP is robust to heavy-tailed distributions, but is susceptible to changes in variance, adding artificial changes in the first 350 observations.
On the contrary, HSMUCE, QS, and NOT(HT) suffer from oversmoothing and miss important data features.
In this case, MQSE is capable of exploring the concentration of points around the median from the heavy-tailed $t$ (1 d.f.) distribution, despite low signal to noise ratio.
NWBS is very robust to different noise families, but adds change-points in every change of distribution function, in this case in the first change from the normal to the t distribution.
Moreover, we found NWBS (and QS) to be several magnitudes slower than the other methods, for instance, for this data example NWBS took $64s$ and MQS $0.64s$.

\begin{figure}
\centering
\includegraphics[width=\textwidth]{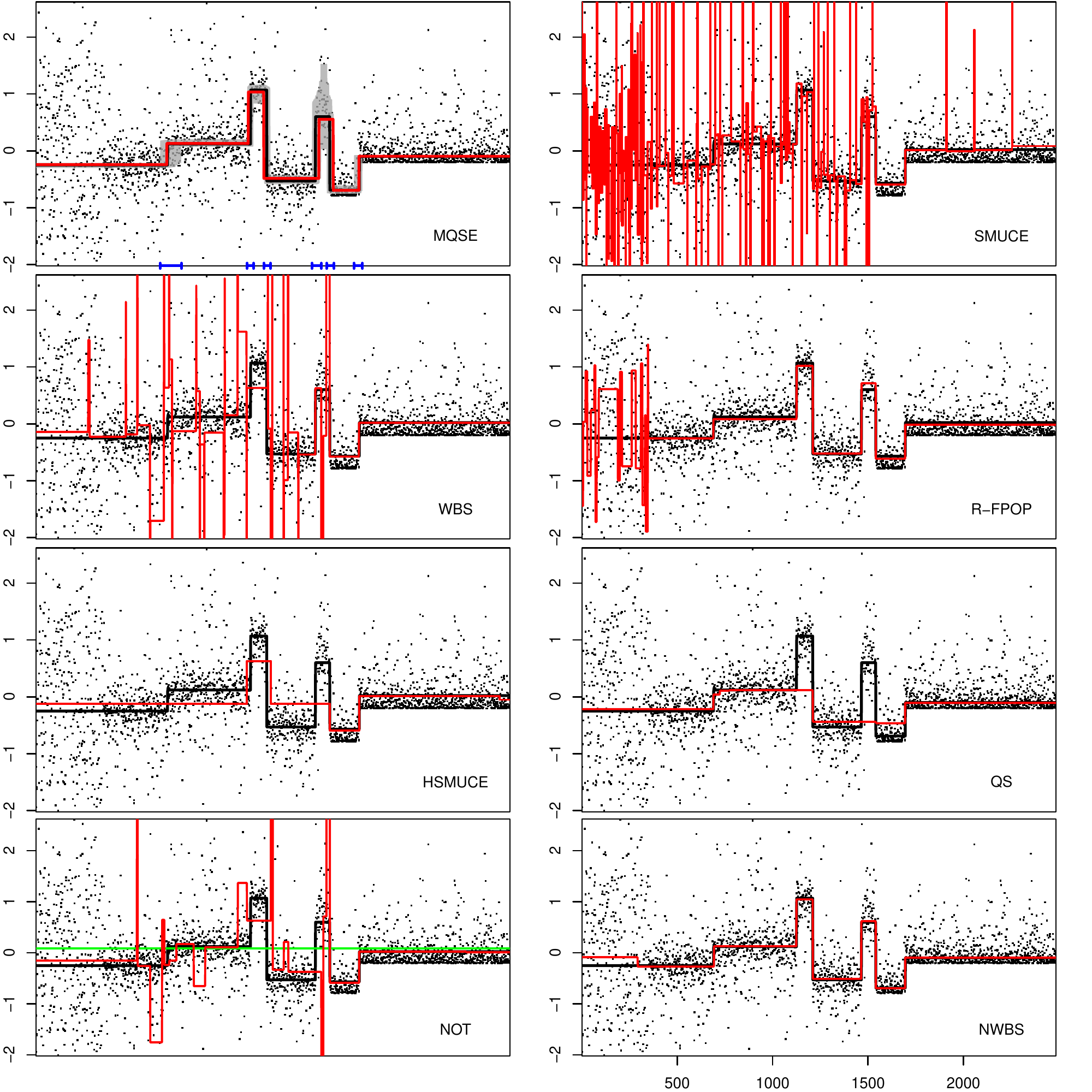}
\caption{\small
Observations (black dots) with different distributional regimes (see details in main text and Appendix \ref{sec:example2}) and true median (black lines).
Left to right, top to bottom (red lines): MQSE for median with confidence band (grey area) and confidence intervals (blue intervals), $\alpha = 0.1$; SMUCE \citep{frick2014}, WBS \citep{fryzlewicz2014}; R-FPOP \citep{fearnhead2017}; HSMUCE \citep{pein2017a}; QS \citep{eilers2005}; NOT (HT, green line) and NOT (VAR, red line) \citep{baranowski2019}; NWBS \citep{padilla2019}.}
\label{fig:rob}
\end{figure}
\begin{figure}
\centering
\includegraphics[width=\textwidth]{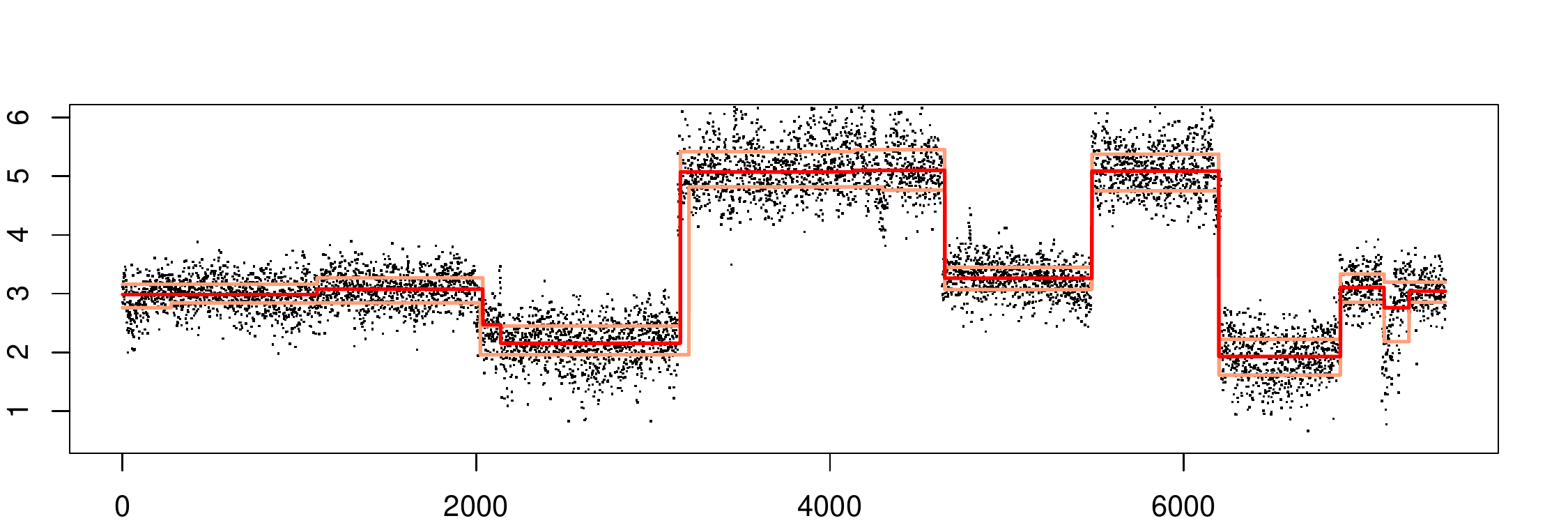}
\caption{ \small Preprocessed WGS data (black dots) of cell line LS411 from colorectal cancer and MQS multiscale box plot with $\alpha = 0.01$ for the underlying CNA's (red and salmon lines). Sequencing was performed by Complete Genomics in collaboration with the Wellcome Trust Centre for Human Genetics at the University of Oxford.}
\label{fig:cnacomp}
\end{figure}

Recall that MQS does not just provide some reconstruction of the underlying $\beta$-quantile, but it also comes with precise statistical guarantees, such as honest confidence statements.
This is particularly valuable for many real data examples, where there is uncertainty about the precise observational distribution.
We demonstrate this in Figure \ref{fig:cnacomp} (see also Figure \ref{fig:wgsmulex}) which shows an example from cancer genetics \cite{behr2018}, where one aims to detect copy-number aberrations (changes in the number of copies of certain regions of the genome) in tumor DNA (see Section \ref{sec:app} for details and a further data set on ion channel recordings).
While Figure \ref{fig:wgsmulex} shows that most of the other methods include artificial jumps in their reconstruction (which in this data example are known to be not present in the underlying signal and result from sequencing artifacts, see Section \ref{subsec:cnaapp} for details), MQS reliably recovers most of the copy-number aberrations correctly.
Moreover, copy-number aberrations are known to often omit heterogeneity on different segments, and MQS adds a powerful visualization tool for this via the multiscale box plot, see Figure \ref{fig:cnacomp}. 

\subsection{Related work}\label{subsec:rel}
This work extends previous methodology developed in \citep{frick2014} on change-point regression under a parametric model (e.g., for a Gaussian mean signal) to a semi-parametric model.
A major and novel feature of our approach is that it comes with statistical guarantees without making any such parametric assumptions.
In particular, our confidence statements for all quantities and minimax optimal estimation rates, expand those offered in \cite{frick2014} for parametric models to quantile segmentation.
The technical challenges in such a semi-parametric model are tackled by a novel exponential inequality (see (\ref{eq:2statg})) based on the quantile jump function (see Definition \ref{def:quantijump}), which is fundamental to our theory.
Algorithmically, we are able to exploit a double heap structures for the optimization problem in (\eqref{eq:hatK}), which overall only leads to a loss of a $\log(n)$ factor in runtime compared to the parametric situation in \citep{frick2014}.
The transformation (\ref{eq:transformtion}) extends the use of residual signs of \cite{dumbgen1998} and \cite{davies2001} and was already mentioned in \citep{frick2014}, however, without any theoretical analysis or efficient implementation.

The estimation of step functions with unknown number and location of segments are widely discussed problems; we mention \citep{ fearnhead2006, spokoiny2009,boysen2009, harchaoui2010, jeng2010, killick2012, niu2012, siegmund2013, matteson2014, du2016, gao2019, fryzlewicz2018} for a variety of methods and related statistical theory.
Methods for segment regression problems that offer statistical guarantees in general restrict those to (minimax optimal) estimation rates which typically concern the mean function and assume normality, see e.g., \citep{harchaoui2010, cai2012, fryzlewicz2014,li2016, pein2017a, baranowski2019}. 
In contrast, we target the quantile function (see e.g., \cite{koenker2005} for a survey) and do not assume a specific distributional model. One scenario related to our work is nonparametric distribution segmentation, see e.g. \citep{zou2014,chu2019,padilla2019}. In this case, the aim is to pick up any distributional changes, not only in a particular quantile. 
Conceptually, mostly related to our work are methods which explicitly aim to provide robust methodology for segmentation.
From a general perspective, segment regression may be considered as a particular case of a high dimensional linear model, with a very specific design matrix.
\cite{belloni2011} consider sparse high dimensional quantile regression and show that minimizing the asymmetric absolute deviation loss together with $L_1$ penalization yields almost optimal estimation rates (in $L_1$ loss).
However, their results require certain regularity conditions on the design matrix (similar to restricted eigenvalue conditions) which are not fulfilled for the specific design which corresponds to a multiscale segmentation setting.
Other methods are concerned with particular types of robustness. 
For example \cite{fearnhead2017, baranowski2019} consider the case of heavy-tailed symmetric distributions.
\cite{eilers2005,li2007} consider minimizing absolute loss with an $L_1$-penalty for quantile regression for the segment regression setting, while not providing any theoretical results.
More generally, \cite{dumbgen2009} consider quantile regression via minimizing a convex loss with a total variation penalty and provide certain consistency results.
Finally, \cite{aue2014,aue2017} consider quantile segmentation in a time series context based on a minimum description length criterion to select the number of segments and show consistency of their method, whereas in our (simpler) setting we obtain exponentially fast selection consistency, see (\ref{eq:conbound}).

\section{Theory}\label{sec:theory}
In Section \ref{subsec:theoMS} we introduce the MQS methodology in more detail and in Section \ref{subsec:theoCons} we present our results on model selection consistency, confidence statements, and minimax optimal estimation rates for MQS. Proofs are postponed to Appendix A. We start with a discussion of the model assumptions.

At first glance, it might seem restrictive that the QSR-model requires for the underlying $\beta$-quantile regression function $\vartheta_\beta$ that $\Pp(Z_i \leq \vartheta_\beta(x_i)) = \beta$ (and not, more generally, that $\vartheta_\beta(x_i) = \inf \{ \theta : \Pp(Z_i \leq \theta ) \geq \beta \}$).
However, such assumption is unavoidable for efficient (i.e., with non-trivial power) control of overestimation for the number of segments, as it is achieved by the proposed method, without assuming a specific observational distribution.
%
To see this, note that without making any distributional assumptions, for given independent observations $Z_1,\ldots,Z_n$ and candidate quantiles $\vartheta(x_1),\ldots ,\vartheta(x_n)$, all information from $Z_i$ about the candidate $\vartheta$ is captured in the transformation $W_i = W_i(Z_i, \vartheta(x_i))$, $i = 1,\ldots,n$ in (\ref{eq:transformtion}).
Now assume that in the QSR-model the condition in (\ref{eq:model1}) is replaced by 
\begin{align}\label{eq:generalQuantile}
\vartheta_{\beta}(x_i) = \inf \{ \theta :\; \Pp(Z_i \leq \theta) \geq \beta \}.
\end{align} 
Consider some (discontinuous) distribution function $F$ for which (\ref{eq:generalQuantile}) and (\ref{eq:model1}) differ, such that there exists $\theta_0 \in \R$ with $F(\theta_0) = 1$ and $\lim_{x\nearrow \theta_0} F(x) < \beta$.
Consider observations in the QSR-model with $Z_1,\ldots,Z_n \overset{i.i.d.}{\sim} F$.
Then, the true $\beta$-quantile $\vartheta_\beta$ is constant with $\vartheta_\beta \equiv \theta_0$ and for any data $Z_i$ the transformation in (\ref{eq:transformtion}) with truth $\vartheta_\beta$ yields $W_i = W_i(Z_i, \theta_0) = 1$ for $i = 1,\ldots,n$.
However, for any other data $Z_i$ the constant candidate $\vartheta \equiv \theta_0$ with $\theta_0 \geq \max(Z_i)$, results in exactly the same transformation $W_i = W_i(Z_i, \theta_0) = 1$ for $i = 1,\ldots,n$.
Consequently, if one allows in the QSR-model for generalized quantiles as in (\ref{eq:generalQuantile}) controlling segment overestimation as in Theorem \ref{theo:overestimation} appears too ambitious, as it rules out any reasonable estimator, i.e., this can only be achieved when setting $S=0$ always.
%
%
%
%
In fact, quantile estimation of non continuous distributions is well known to require more specific model assumptions in general, see e.g.\ \citep{machado2005}. The MQS procedure allows quantile regression for discrete distributions, as long as the assumptions of the QSR-model hold.

\subsection{Multiscale procedure}\label{subsec:theoMS}
Given $\beta$, observations $Z=(Z_1,\ldots,Z_n)$ from the QSR-model and a function $\vartheta\in \cpset$ with $S$ segments as in (\ref{eq:Sset}) (which may depend on $n$), the MQS methodology is based in a first step on a multiscale test to decide whether or not $\vartheta$ is a good candidate for the underlying unknown quantile function $\vartheta_\beta$. 
To this end, fix an interval $[x_i,x_j]\subseteq[\tau_{s-1},\tau_{s})$ for some $1\leq s\leq S$, that is, an interval where the candidate $\vartheta$ is constant with value $\theta_s$. 
When the true quantile function $\vartheta_\beta$ is constant on $[x_i,x_j]$ as well, the decision problem of whether or not $\vartheta$ coincides with the truth on the interval $[x_i,x_j]$ translates to the testing problem 
\begin{align}
\begin{split}\label{eq:1loctest}
H_{ij}  :\mathbb P(Z_k\leq \theta_s)=\beta \quad  \forall \; i\leq k\leq j \quad \text{vs.} \quad
K_{ij}  :\mathbb P(Z_k\leq \theta_s)\neq\beta \quad \forall \; i\leq k\leq j.
\end{split}
\end{align}
Equivalently, one can write (\ref{eq:1loctest}) using the transformed data in (\ref{eq:transformtion}) as the hypothesis testing problem in (\ref{eq:introLocHyp}).
The log-likelihood-ratio test for the hypothesis testing problem (\ref{eq:introLocHyp}) and (\ref{eq:1loctest}), respectively, is given by the test statistic 
\begin{align}
\begin{aligned}
T_i^j(W) &= T_i^j\left(W\left(Z,\theta_s\right)\right)=
\log \left(\frac{\sup_{\beta^\prime}\prod_{l = i}^j f_{\beta^\prime}(W_l)}{\prod_{l = i}^j f_\beta(W_l)} \right)\\
 &= (j-i+1)\left(\overline{W}_i^j\log\left(\frac{\overline{W}_i^j}{\beta}\right)+\left(1-\overline{W}_i^j\right)\log\left(\frac{1-\overline{W}_i^j}{1-\beta}\right)\right),\\
\end{aligned} 
 \label{eq:lrs}
\end{align}
where $f_\beta(x) = \beta^x (1-\beta)^{(1-x)}$ denotes the probability mass function of the Bernoulli distribution and $\overline{W}_i^j \ZuWeis (j-i+1)^{-1}\sum_{l=i}^jW_l$.
That is, the test is then of the form
\begin{equation}\label{eq:localTests}
\Phi_{ij}\left(Z\right)=
\begin{cases}
 0& \mbox{ if } \sqrt{ 2 T_i^j(W(Z,\theta_s))}\leq q_{i,j},\\
 1&\mbox{ otherwise,}
\end{cases}
\end{equation}
where the threshold $q_{i,j}:=q_{i,j}(\alpha)$ determines the level $\alpha$ of the test (note that, under $H_{ij}$ the distribution of $T_i^j(W)$ is independent of $\vartheta_\beta$).
However, we do not know where and on which scales the true $\vartheta_\beta$ is constant and 
thus, we have to consider all intervals on all different scales simultaneously.
This means that a candidate function $\vartheta$ is accepted if and only if it gets accepted by all local tests $\Phi_{ij}$ in (\ref{eq:localTests}), for appropriately chosen thresholds $q_{ij}$.
%
%
%
%
More precisely, define the penalized multiscale statistic (as a functional on $\cpset$) as
\begin{equation} T_n(Z,\bullet)\ZuWeis \max_{\substack{1\leq i\leq j\leq n\\ \bullet \mbox{ is constant on }[x_i,x_j]}}\sqrt{2T_i^j(W(Z,\bullet),\beta)}-\pen.\label{eq:mts}
\end{equation}
with 
$\pen \ZuWeis \sqrt{ 2 \log\left( {\textup{e} \, n}/{\ell}\right)}$,
where $\ell=j-i+1$ denotes the scale and $\textup{e} = \exp(1)$.
For some given threshold $q = q_n(\alpha)$, a candidate function $\vartheta$ is accepted if and only if $T_n(Z,\vartheta) \leq q_n(\alpha)$. This means that we choose local thresholds $q_{ij}(\alpha)$ in (\ref{eq:localTests}) as in \citep{dumbgen2001,dumbgen2008, frick2014} of the form 
\begin{equation}\label{eq:qij}
q_{ij}(\alpha) = q_n(\alpha) +  \pen.
\end{equation}
%
%
\begin{remark}
A heuristic reasoning for the particular penalization $\pen$ in (\ref{eq:mts}) is as follows.
%
%
In order to put different scales on equal footing, one has to chose larger thresholds for smaller scales. 
On the other hand, it follows from Wilk's theorem (see \citep{frick2014} for a more precise argument) that for sufficiently large intervals the local statistics $\sqrt{2 T_i^j}$ are well approximated by the absolute value of a standard normal. 
%
%
As the maximum of $\ell/n$ independent standard normals grows as $\sqrt{2 \log(\ell/n)}$, the particular choice in (\ref{eq:qij}) ensures that $T_n(Z,\bullet)$ in (\ref{eq:mts}) remains finite even when $n \to \infty$ for any fixed $\alpha \in (0,1)$.
\end{remark}

Note that for $X=(X_1,\ldots,X_n)$ i.i.d. Bernoulli distributed with mean $\beta$ and $\vartheta = \vartheta_\beta$, the statistic $T_n(Z,\vartheta_\beta)$ in (\ref{eq:mts}) follows the same distribution as
\[\max_{\substack{1\leq i\leq j\leq n\\ \vartheta_\beta\mbox{ is constant on }[x_i,x_j]}}\sqrt{2T_i^j(X,\beta)}-\pen,\]
%
%
%
which
is bounded by $M_n$ in (\ref{eq:mn}) in probability.
The distribution of $M_n$ does not depend on any characteristics of the unknown $\vartheta_\beta$ and hence, we can determine its quantiles via Monte-Carlo simulations. 
Thus, we can define $q_n(\alpha)$ as the $(1-\alpha)$-quantile of $M_n$, i.e.,
\begin{equation}
q_n(\alpha) \ZuWeis \inf\{q\; : \; \mathbb{P}(M_n\leq q) \geq 1-\alpha\}.\label{eq:quantile}
\end{equation}
This choice guarantees that the true quantile function $\vartheta_\beta$ gets accepted by the above testing procedure with probability at least $1-\alpha$, which leads to confidence statement for all quantities of $\vartheta_\beta$, see Theorem \ref{theo:confidenceset}.

\begin{remark}\label{rem:limitDistM}
It is shown in \citep{frick2014, dumbgen2001} that $M_n$ converges in distribution to an almost surely finite random variable $M$ (which is independent of $\beta$) and hence, $\limsup_n q_n(\alpha) < \infty$ for any $\alpha \in (0,1)$ . 
In particular, for large $n$, one may choose the parameter $q(\alpha)$ as a quantile of $M$, which can be simulated once and then stored, lowering the overall computation time.
\end{remark}
Finally, we define the MQS estimator (MQSE) $\hat{\vartheta}$ as an element in $\cH(q_n(\alpha))$ from (\ref{eq:H}) that minimizes a particular cost function.
One cost function that has also been considered in several other works \citep{eilers2005,li2007,dumbgen2009,belloni2011,lee2018} is the asymmetric absolute deviation loss. 
That is, 
\begin{equation}\hat{\vartheta} \coloneqq\argmin_{\hat{\vartheta}\in\mathcal{H}(q)}\sum_{i=1}^n\left(Z_i-\hat{\vartheta}(x_i)\right)\left(\beta-\mathbbm{1}_{\{Z_i<\hat{\vartheta}(x_i)\}}\right).\label{eq:koenker}
\end{equation}
In Appendix \ref{subsec:theoChoice}, we give details on an alternative choice which is based on the Wald-Wolfowitz runs statistic \citep{wald1940}.
We stress that all our theoretical results, detailed in the following, hold for any estimator in $\cH(q_n(\alpha))$ and are thus, completely independent of this choice.
Both cost functions, the one in (\ref{eq:koenker}) and in Appendix \ref{subsec:theoChoice}, are available in our R package implementation.
In simulations in Appendix \ref{subsec:theoChoice} we find that they typically perform comparably, with the latter slightly outperforming (\ref{eq:koenker}) for higher quantiles. 

\subsection{Consistency results and confidence statements}\label{subsec:theoCons}
From the construction of the MQSE $\hat{\vartheta}$ it follows with $q_n(\alpha)$ as in (\ref{eq:quantile}), that the corresponding number of segments does not exceed the true number of segments at given error level $\alpha$, as described in equation (\ref{eq:1statg}) and Theorem \ref{theo:overestimation}.
%
%
As already argued in Section \ref{sec:intro}, any bound on the underestimation error must depend on further characteristics of the underlying quantile function $\vartheta_\beta$ and on the respective quantile jumps. To this end, let
\begin{equation}\label{eq.lambdaXi}
\xi_s \ZuWeis \min\left\{\xi_{F_{\min},\beta}(-\delta_s),\xi_{F_{\max},\beta}(\delta_s)\right\} \quad s = 1,\ldots,S, 
\end{equation}
with $\delta_s=\theta_{s}-\theta_{s-1}$, $F_{\min}\ZuWeis \min\{F_1,\ldots,F_n\}$, and $F_{\max}\ZuWeis \max\{F_1,\ldots,F_n\}$ and
\begin{equation}\label{eq:gamma}
\gamma_{n,s}(q)=\left(1-2\exp\left(-\frac{\left(\sqrt{2n\lambda_s}\,\xi_{s}-q-\sqrt{2\log(2\e/\lambda_s)}\right)_+^2}{2}\right)-2\exp\left(-n\lambda_s\,\xi_{s}^2\right)\right)^2.
\end{equation}

\begin{theorem}[Underestimation error]\label{theo:underestimation}
Consider the QSR-model and $\hat{S}(q)$ as in (\ref{eq:hatK}). Then, for $q>0$ 
\[\mathbb{P}\left(\hat{S}(q)\geq S\right)\geq \prod_{s=1}^{S}\gamma_{n,s}(q).\]
\end{theorem}
From Theorem \ref{theo:underestimation} we obtain that for any fixed $q>0$ and $\Lambda,\Xi > 0$ in (\ref{eq:LambdaXi}) the probability of underestimating the number of segments vanishes exponentially fast, see Theorem $\ref{cor:easyunder}$. Additionally, we obtain from Remark \ref{rem:limitDistM} and Theorem \ref{theo:underestimation} for $\Lambda_0>0$ and $\Xi_0>0$, that for any $\alpha \in (0,1)$ as $n\rightarrow\infty$
\begin{align*}
   \inf_{\substack{\vartheta\in\cpset: \\ \Lambda>\Lambda_0,\Xi
   >\Xi_0 }} \mathbb{P}\left(\vartheta_\beta\in\mathcal{H}(q_n(\alpha))\right)&\geq \mathbb{P}\left( T_n(Z,\vartheta_\beta) \leq q_n(\alpha) \right)-\mathbb{P}\left(\hat{S}(q)<S\right)\\
    &\geq 1-\alpha - \mathbb{P}\left(\hat{S}(q)<S\right)\geq 1-\alpha+o(1)
\end{align*}
and, therefore, we can state the following theorem.
\begin{theorem}[Confidence statements]\label{theo:confidenceset}
Consider the QSR-model and $\hat{S}(q)$ as in (\ref{eq:hatK}), $q_n(\alpha)$ as in (\ref{eq:quantile}), $\Lambda$ and $\Xi$ as in (\ref{eq:LambdaXi}). Then, the set $\cH(q_n(\alpha))$ in (\ref{eq:H}) constitutes an asymptotically honest $(1-\alpha)$-confidence band for $\vartheta_\beta$ in the QSR-model uniformly over segment functions with minimal $\Lambda_0, \Xi_0 > 0$, i.e.,
\[\lim_{n\rightarrow\infty}\inf_{\substack{\vartheta\in\cpset: \\ \Lambda>\Lambda_0,\Xi
   >\Xi_0 }} \mathbb{P}(\vartheta\in\cH(q_n(\alpha)))\geq 1-\alpha.\]
   
\end{theorem}
As detailed in Appendix \ref{sec:detailsImplementation}, $\cH(q)$ can be computed easily simultaneously with the MQS estimator $\hat{\vartheta}$ and from $\cH(q_n(\alpha))$ confidence intervals for the segment locations $\tau$ and for the quantile values $\theta$ can be constructed.

From Theorems \ref{cor:easyunder} and the fact that if $q_n \to \infty$, $\alpha_n \ZuWeis \Pp(M_n > q_n) \to 0$, it follows directly that for any fixed $\Lambda, \Xi > 0$ and some sequence $q_n \to \infty$ such that $\sfrac{q_n}{\sqrt{n}}\rightarrow 0$ MQS performs consistent model selection for the number of c.p.'s.
The following result goes beyond this and considers the situation of a sequence of regression function $\vartheta_\beta(n)$ in the QSR-model, where the minimal scale $\Lambda = \Lambda_n$ and the minimal quantile jump $\Xi = \Xi_n$ can vanish as $n \to \infty$.
\begin{theorem}[Model selection consistency]\label{theo:vanishingrates}
For a sequence $\vartheta_{\beta,n} \in \cpset$ with $\Lambda_n$ and $\Xi_n$ as in (\ref{eq:LambdaXi}), consider the QSR-model and $\hat{S}(q)$ as in (\ref{eq:hatK}). For some sequence $q_n \to \infty$, assume the following.
\begin{enumerate}
    \item For signals with $\liminf_n\Lambda_n>0$ and $\liminf_n\Xi_n>0$, that $\sfrac{\sqrt{n}}{q_n}\rightarrow \infty$.
    \item For signals with $\liminf_n\Lambda_n>0$ and $\Xi_n\rightarrow 0$, that $\sfrac{\sqrt{n}\Xi_n}{q_n}\rightarrow\infty$.
    \item For signals with $\Lambda_n\rightarrow 0$, that $\sqrt{n\Lambda_n}\Xi_n\geq(2+\epsilon_n)\sqrt{-\log(\Lambda_n)}$, for some $\epsilon_n$ with $\epsilon_n{{\sqrt{-\log\Lambda_n}}/{q_n}\rightarrow\infty}$.
\end{enumerate}
Then, for such sequences $\vartheta_{\beta,n}$, the MQS selector is consistent, that is,
$\mathbb{P} \left(\hat{S}(q_n) = S \right)\rightarrow 1$.
\end{theorem}
Theorem \ref{theo:vanishingrates} shows that for a sequence of signals $\vartheta_{\beta,n}$ the number of segments is estimated consistently as long as for the respective minimal scale $\Lambda_n$ and minimal quantile jump $\Xi_n$ it holds that
\begin{equation}\label{eq:minimaxDB}
\sqrt{n \Lambda_n} \Xi_n > 2 \sqrt{- \log(\Lambda_n)}.
\end{equation}
\cite{frick2014} showed that in the case of Gaussian observations with piecewise constant mean, no method can consistently estimate the number of segments for a sequence of signals with minimal scale $\Lambda_n$ and minimal jump height $\Delta_n$ whenever $\sqrt{n \Lambda_n} \Delta_n <  \sqrt{- 2 \log(\Lambda_n)}+o(1).$
Furthermore, for any continuous distribution $F$ with density $f$ we find that $\lim_{\delta \to 0} \frac{\xi_{F,\beta} (\delta)}{\delta} = f(\theta_{\beta})$.
For Gaussian observations, the mean and the median coincide, hence, one obtains that $\Xi_n \in \mathcal{O}(\Delta_n)$. 
Consequently, possibly up to the constants, (\ref{eq:minimaxDB}) cannot be improved in general. 
Besides model selection consistency and confidence statements for all quantities, MQS also yields minimax optimal estimation rates for the location of segments (up to a log-factor) as the following theorem shows.
\begin{theorem}[Estimation rates]\label{theo:cprate}
Consider the QSR-model, $\hat{\vartheta} \in \cH(q)$ as in (\ref{eq:H}), and $\Lambda,\Xi$ as in (\ref{eq:LambdaXi}). Then, for any $q>0$ and sequence $\epsilon_n\searrow 0$
\[\mathbb{P}\left(\max_{\tau\in J(\vartheta)}\min_{\hat{\tau}\in J(\hat{\vartheta}(q))}|\hat{\tau}-\tau|>\epsilon_n\right)\leq 4 (S-1) \textup{e}^{-n\epsilon_n\Xi^2}\left[\textup{e}^{2\sqrt{n\epsilon_n}\,\Xi\left(\sfrac{q}{\sqrt{2}}+\sqrt{\log(2\textup{e}/\epsilon_n)}\right)}+1\right].\]
\end{theorem}
Note that for $\Lambda^{-1} = o\left(n/\log(n)\right)$ and $q = o\left(\sqrt{\log(n)}\right)$ a sufficient condition for the right hand side to vanish as $n\rightarrow\infty$ is $\epsilon_n\geq \log n/ (\Xi^2 n)$,
which, up to the log-term equals the minimax optimal sampling rate $1/n$ under a normal error assumption.

\begin{remark}[Consistency for increasing number of change-points]
We stress that for all our results, the number of change-points $S$ can depend on $n$, that is, we can consider a sequence $(S_n)_{n\in\mathbb{N}}$.
Note that, the number of change-points $S_n$  is upper bounded by the minimal scale $\Lambda_n$ via $\Lambda_n^{-1}\geq S_n-1$ for all $n\in\mathbb{N}$.
As argued above, our MQSE remains consistent when $\Lambda_n$ converges to $0$ at rate  $\Lambda^{-1} = o\left(n/\log(n)\right)$, i.e., even when the number of change-points $S_n$ goes to infinity as $n$ increases, as long as  $S_n \in o\left(n/\log(n)\right)$.
\end{remark}

\section{Implementation }\label{sec:implementation}

The MQSE and its associated confidence bands can be computed with dynamic programming, employing the double heap structure underlying the computation of quantiles as in \citep{astola1989}. 
The major idea of the dynamic programming scheme is analog to the one in \citep{frick2014}:
One successively computes the segmentation (including confidence statements) for the sub-problems consisting of the first $j$ data points, with $j = 1, \ldots, n$.
Thereby, one makes use of overlapping structure, such that the $j$th solution is updated efficiently from the solutions of the first $j - 1$ sub-problems.
The basic Bellman equation in this dynamic programming scheme is as follows:
For the $j$th sub-problem with observations $Z_1, \ldots, Z_j$, if one knew that the last change-point of the MQSE is at location $i$, then the MQSE for the $j$th sub-problem equals the MQSE for the $i$th sub-problem on the first $i$ data points.
There are $\mathcal{O}(n)$ possible locations for the  \textit{last change-point} and there are $n$ successive sub-problems.
Thus, overall this dynamic programming scheme has a worst case complexity of $\mathcal{O}(n^2)$ times the cost to compute the constant solution between the last change-point $i$ and the last data point $j$ on interval $[x_i, x_j]$.
Note that, in order to compute a constant MQS solution on some interval $[x_i, x_j]$, one has to intersect all confidence intervals associated with the level $\alpha$-tests $\phi_{k,l}$ in (\ref{eq:localTests}) for $i \leq k \leq l \leq j$ .
When the intersection is empty, this means that there is no constant function for which all local tests in (\ref{eq:localTests}) accept and hence, a new change-point has to be added. 
For sake of brevity, here we only outline the major differences to the algorithm in \citep{frick2014}.
More details on the implementation, including confidence statements and the specific cost function for the MQSE, are given in Appendix \ref{sec:detailsImplementation}.
For the MSB, which plots estimates for the $0.25$, $0.5$ and, $0.75$ quantiles simultaneously, in order to avoid crossing of different quantiles, we make some further minor modifications detailed in Appendix \ref{sec:mbp}.

The local level $\alpha$-tests in (\ref{eq:localTests}) can be inverted into $(1-\alpha)$-confidence statements for the underlying parameter $\theta_s$ in (\ref{eq:1loctest}) as follows.
For $q>0$ define $l(q)$ and $u(q)$ as the two unique solutions of 
\begin{equation}x\log\left(\frac{x}{\beta}\right)+(1-x)\log\left(\frac{1-x}{1-\beta}\right)=q\label{eq:solve}
\end{equation}
such that $0\leq l(q) < u(q) \leq 1$. 
Then, some straight forward calculations show that $T_n(Z,\vartheta) \leq q$ if and only if
\begin{align}\label{eq:boxBij}
\vartheta|_{[x_i,x_j]} \in \left[ Z^{i,j}_{[\underline{m}_{i,j} ]},Z^{i,j}_{[\overline{m}_{i,j}+1 ]}\right)\WeisZu \left[ \underline{b}_{ij}(q), \overline{b}_{ij}(q) \right), \forall i,j \text{ with } \vartheta|_{[x_i,x_j]}  \text{ constant,}
\end{align}
with $\underline{m}_{i,j}= \max \left( 1, \lceil (i-j+1)l(\tilde{q}) \rceil \right) $, $\overline{m}_{i,j}= \min \left( n, \lfloor (i-j+1)u(\tilde{q}) \rfloor \right)$ and $\tilde{q} = (q + \pen^2 / (2 (j-i+1))$. 
Note that $\underline{m}_{i,j} = \underline{m}_{1,j-i+1} \WeisZu \underline{m}_{j-i+1}$ and $\overline{m}_{i,j} = \overline{m}_{1,j-i+1} \WeisZu \overline{m}_{j-i+1} $ only depend on the length of interval $[x_i,x_j]$.
Just as in \citep{frick2014}, the computation of MQS is based on these confidence boxes $\left\{\left[ \underline{b}_{ij}, \overline{b}_{ij}\right) \; : \; 1\leq i\leq j \leq n\right\}$ in (\ref{eq:boxBij}).
In \citep{frick2014} the boxes $\left[ \underline{b}_{ij}, \overline{b}_{ij}\right)$ depend on the local sums of observations $\sum_{l = i}^j Z_l$ and its dynamic program explores that these local sums and hence, boxes of an interval $[i,j]$, can be updated in $\mathcal{O}(1)$ time from the boxes of intervals $[i+1,j]$ and $[i,j-1]$, respectively.
In contrast, for MQS the boxes depend on the particular quantiles of the observations $\{Z_i,\ldots,Z_j\}$. 
Thus, in order to adapt the programming scheme of \cite{frick2014}, one has to update the running $\overline{m}_1,\underline{m}_1,\ldots,\overline{m}_n,\underline{m}_n$-quantiles efficiently.
Here, we use double heap structures as in \citep{astola1989} to update the running quantile in $\mathcal{O}(\log(n))$ time.
In total, this increases the overall computation time by a log-factor, with worst case complexity of order $\mathcal{O}(n^2 \log(n))$.
However, depending on the reconstructed signal, pruning steps and the use of smaller interval systems often lead to a computation time which is almost linear in $n$, see \cite{frick2014} for details.

\section{Simulations}\label{sec:simulations}

In the following, we explore MQS in a simulation study.
Thereby, the choice of the threshold parameter $q$ is essential as it balances detection and overestimation of the number of segments and hence false positives.
Thus, $q$ can be seen as a tuning parameter of MQS.
Via the one-to-one correspondence of a confidence level $\alpha$ and $q_n(\alpha)$ in (\ref{eq:quantile}), in the following, we choose $\alpha=0.1$ and hence $q = q_n(0.1)$.
In this way, the probability that MQS overestimates the number of segments does not exceed $10\%$. Such a choice depends on the application. However, we see a great advantage of our methodology, to rely only on this parameter, which has an immediate statistical meaning.
For a more refined discussion on the choice of threshold parameter $q$ and possible data driven model selection procedures, we refer to \citep{frick2014}.
For example, another possible parameter choice for $q$ is via minimizing the right hand side of  (\ref{eq:conbound}) together with Monte-Carlo simulations of $M_n$. 

In the following, we consider five different competitors for MQS for which software is available online: 
SMUCE from \citep{frick2014}, HSMUCE from \citep{pein2017a}, wild binary segmentation (WBS) from \citep{fryzlewicz2014}, R-FPOP from \citep{fearnhead2017}, quantsmooth (QS) from \citep{eilers2005}, NOT from \citep{baranowski2019}, and NWBS \citep{padilla2019}.
SMUCE provides a multiscale methodology for normal observations with homogeneous variance.
HSMUCE is also designed for normally distributed data, but is robust against changes in variance.
WBS can be seen as a ``greedy" procedure which successively adds changes based on a localized CUMSUM statistic. The theoretical results implicitly assume normally distributed observations with change in mean. 
R-FPOP is a penalized cost approach that uses the biweight loss, which is designed to be particularly robust to extreme outliers. It considers arbitrary changes in the underlying distribution, but cannot be tuned to search for changes in some specific quantiles. In particular, it considers that observations in segments  are i.i.d..
QS is a smoothing method which performs minimization of the asymmetric absolute deviation loss (recall (\ref{eq:koenker})) together with $L_1$ penalization and hence, can compute arbitrary quantile curves. However, QS is not designed for change-point detection, but for data visualization through smoothing. This will be clear in simulations, where the number of segments is very often very large.
NOT choose random subsamples and then uses a tailor-made contrast function to find possible change-points. The contrast functions more related to our work are the (HT) for heavy-tailed noised and (VAR) for changes in variance.
NWBS is a non parametric method based on a CUSUM-like statistic for changes in the distribution function. It then uses wild binary segmentation ideas, as in WBS.
For all competitors, we choose tuning parameters as default in the available software.
%
%
%
%
We summarize some features of these methods in Table \ref{tab:summarytable}.

As measures of evaluation for the simulation study we use the number of estimated segments, the mean absolut squared error (MIAE) $\sum_{i = 1}^n |\hat{\vartheta}(x_i) - \vartheta_\beta(x_i)| / n$, and the entropy-based V-measure introduced in \citep{rosenberg2007}. The latter takes values in $[0, 1]$ and measures whether given clusters include the correct data points of the corresponding class. Larger values indicate higher accuracy with $1$ corresponding to a perfect segmentation.
All results were obtained from $1,000$ Monte Carlo runs.
We also performed simulations on the coverage properties for the confidence intervals and bands. 
Details can be found in Appendix \ref{sec:conf}. 
Similar to other multiscale approaches \citep{frick2014,behr2018}, we find that the level is typically exceeded, meaning that, in general, these confidence statements are conservative.
\subsection{Additive error}\label{subsec:simplesec}
First, we consider an additive model with i.i.d.\ error terms, that is,
\begin{equation}
Z_i=\vartheta(x_i)+\varepsilon_i \quad\mbox{for }i=1,\ldots,n\label{eq:simple}
\end{equation}
with $\vartheta\in\cpset$ and $\varepsilon_1,\ldots,\varepsilon_n$ i.i.d.\ according to some distribution. 
For observations as in (\ref{eq:simple}) in the QSR-model the quantile functions $\vartheta_\beta$ are shifted versions of $\vartheta$, namely,
$\vartheta_\beta=\vartheta+\theta_\beta$,
where $\theta_\beta$ is the $\beta$-quantile of $\epsilon_1$. 
In particular, for any quantile $\beta \in (0,1)$, $\vartheta_\beta$ has the same number and locations of segments.
Here, we consider $\vartheta$ as in Figure \ref{fig:bigex} (top row), which has $7$ segments and $n = 1,988$. 
For the error terms $\varepsilon_i$ we consider normal distribution $\varepsilon_i \sim \cN(0, \sigma^2)$, $t$-distribution with $3$ degrees of freedom $\epsilon_i \sim t_3\sigma / \sqrt{3}$ and variance $\sigma^2$, rescaled Cauchy distribution (heavy tails) $\varepsilon_i \sim 0.02\,\text{Cauchy}(0,1)$, and rescaled chi-square distribution (skewed) with $3$ degrees of freedom and median $0$, that is, $\varepsilon_i \sim (\chi_3^2-\beta_{0.5}(\chi_3^2))\sigma/\sqrt{6}$ with variance $\sigma^2$, with $\sigma^2 = 0.04$ and $\beta_{0.5}(\chi_3^2)$ the median of the distribution $\chi_3^2$. (see Figure \ref{fig:bigex}).
%
%
%

%
The results for MQS$(\beta)$, $\beta = 0.25, 0.5, 0.75$ are shown in Table \ref{tab:simpletable}. 
Note that, in general, the MQS selector of $S$ for $\beta = 0.5,0.75$ seems to have a higher detection power as for $\beta = 0.25$ in this example. 
This is due to the fact that the test signal $\vartheta$ in Figure \ref{fig:bigex} (top row) has 4 jumps upwards but just 2 jumps downwards.
It is easy to check that jumps upwards have a stronger influence on higher (overall) empirical quantiles.
MQS is a reasonable estimator in the four scenarios presented in this section. It is robust against outliers, as well as to skewness of the distributions. 
For normally distributed data is not surprising that methods tailored for this scenario outperform MQS, such as SMUCE, HSMUCE, WBS, and NOT(VAR), in particular in terms of MIAE and segmentation accuracy.
However, MQS(0.5) is comparable in terms of calculating the number of segments. 
On the other hand, for heavy tailed and skewed distributions SMUCE, WBS, and NOT(VAR) fail completely, while MQS(0.5) retains a very high segmentation accuracy, also outperforming HSUMCE for skewed data.
The nonparametric methods tailored for heavy-tails, R-FPOP, NOT(HT), and NWBS, perform comparably to MQS in this additive i.i.d.\ setting.
Although, NOT(HT) highly overestimated the number of segments for Cauchy noise.
%
%
However, as shown in Section \ref{subsec:variancesec}, these methods are highly sensitive to any distributional changes (not present in this set-up), in contrast to MQS which robustly estimates quantile segments among all scenarios.
QS performs comparably to MQS in terms of MIAE, however, it highly overestimates the number of segments and performs worse in terms of estimation accuracy and computational speed. 
This is partly due to the fact that it is designed as a smoothing method and not a segmentation method.

\subsection{Changes in variance}\label{subsec:variancesec}
The QSR-model further allows for changes in variance or other characteristics which are independent of changes in the respective $\beta$ quantile $\vartheta_\beta$. 
Here, we consider changes in variance for normally and $t$ distributed (with $3$ d.f.) observations, $n = 2,000$.
In Figure \ref{fig:varianceex} (top row) the underlying median (mean) (solid line) and variance (dashed line) functions are displayed.
The second and third row show the true  $0.25$, $0.5$, and $0.75$ quantile functions (black lines) and the MQS box plot (red lines).
Note that in this example, the $0.5$ quantile has $4$ and the $0.25$,$0.75$ quantiles have $6$ segments.

Simulation results are shown in Table \ref{tab:variance}. 
MQSE appears very robust against changes in variance within a segment even for heavy tailed distributions.
It estimates the correct number of change-points with high probability.
At the same time, the MQSE for the $0.25$ and $0.75$ quantiles depict changes in variance. 
Similar as before, methods designed for Gaussian (homoscedastic) distributions, such as SMUCE and WBS fail completely for the heavy tailed setting.
Note that, although, HSMUCE was particularly designed for Gaussian distributions with change in variance, MQS(0.5) outperforms it in terms of estimated number of change-points.
As expected, nonparametric methods, such as NWBS and R-FPOP, are sensitive to these changes in variance and perform poorly compared to MQS in terms of estimated number of segments.
Just as in the previous section, QS highly overestimates the number of segments.
NOT(HT) performs comparably to MQS(0.5) in both settings, being better in the normal case, but worse for t distribution.
NOT(VAR) performs bad in both scenarios.
In summary, we find that, whereas other methods only work well in specific situations, MQS is very robust to all scenarios, estimating quantile segments with high accuracy independent of any distributional assumptions.

\section{Real data examples}\label{sec:app}

\subsection{Copy Number Aberrations}\label{subsec:cnaapp}
Copy Number Aberrations (CNA's) are sections of DNA in the genome of cancer cells that are either multiplied or deleted, relative to the state present in normal tissue. 
CNA's are important factors of tumor progression, through the deletion of tumor suppressing genes and the multiplication of genes involved for example in cell division. 
The number of copies of DNA sections, depending on chromosomal loci, corresponds to a segment function, where a segment corresponds to a different copy number.
A common measurement technique is via whole genome sequencing (WGS).
Thereby, the tumor DNA is fragmented into several pieces. 
Then the single pieces are sequenced using short ``reads", and finally these reads are aligned to a reference genome by a computer.
Statistical modeling of WGS data is particularly difficult as random variations and systematic biases, such as mappability and CG bias, lead to violations of parametric model assumptions, such as normal or Poisson, see e.g.\ \citep{liu2013}.
Quantile segmentation with MQS does not require any such specific model assumptions and hence, is particularly suited for this setting.

Figure \ref{fig:cnacomp} shows (pre-processed\footnote{Sequencing produces spatial artifacts in the data and waviness, which can be pre-process using standard procedures of smoothing filter, baseline correction and binning, see, \citep{behr2018} for details.}) WGS data of cell line LS411 from colorectal cancer. Sequencing was performed by Complete Genomics in collaboration with the Wellcome Trust Centre for Human Genetics at the University of Oxford.  
For this particular data set, it is known that the underlying CNA's only take values in the natural numbers. This is because it was collected under special conditions, where cells come from a single homogeneous tumor-clone, see \citep{behr2018} for more details.
This allows certain validation of the estimated segments, something which is not feasible for most real patient tumors.

Figure \ref{fig:cnacomp} shows the MSB (MQS' estimated $0.25$, $0.5$, and $0.75$ quantiles) at confidence level $1- \alpha = 0.99$. 
%
MQSE recovers most of the signal structure correctly.
In particular, MQSE is way more robust than SMUCE \citep{frick2014}, WBS \citep{fryzlewicz2014}, R-FPOP \citep{fearnhead2017}, NOT(HT), and NOT(VAR) \citep{baranowski2019}, see Figure \ref{fig:wgsmulex}.
SMUCE, WBS, R-FPOP, NOT introduce many artificial changes, which cannot be present in the signal as, in this particular example, it is known to only take integer values.
HSMUCE \citep{pein2017a} and NWBS \citep{baranowski2019} 
(fifth and tenth rows in Figure \ref{fig:wgsmulex})
are more robust, but still adds artificial changes.
The sixth row of Figure \ref{fig:wgsmulex} shows the estimated $0.25$, $0.5$, and $0.75$ quantiles of QS \citep{eilers2005}. 
Similar to MQS, it correctly recovers most of the signal structure, but it misses some underlying changes, see, in particular, the last change at data point $7148$.
Moreover, QS has a much higher running time compared to MQS: while MQS took 31 seconds to run each quantile for this data set, QS took 54 minutes. 

\subsection{Ion channel data}\label{subsec:ion}
Ion channels are pore-forming proteins that allow ions to pass through a cell membrane. 
They are vital for several processes like excitation of neurons and muscle cells. 
The pores of an ion channel can open and close, a process called \textit{gating}, often as a result of external stimuli. 
Therefore, the amount of ions that can pass through a channel is not constant in time \citep{chung2007}. 
A major tool for a quantitative analysis of the gating dynamics is the \textit{patch clamp} technique, which allows to measure the conductance of a single ion channel in time \citep{sakmann1995b}. 
Roughly speaking, this kind of data is obtained by inserting a single ion channel in an (often artificial) membrane surrounded by an electrolyte with an electrode to measure the current while constant voltage is applied.
These recordings can be modeled as a segment function disturbed by an error, see e.g.\ \citep{pein2017a, gnanasambandam2017}.
The particular data set considered in Figure \ref{fig:sex} comes from a single channel of the bacterial porin PorB from the Steinam lab (Institute of Organic and Biomolecular Chemistry, University of G\"ottingen). 
%
The measurement protocol incorporates a lowpass filter which leads to local dependencies of the error terms, see \citep{pein2017a}. 
To remove these dependencies, which violate modeling assumptions of the QSR-model, we subsampled every 11th observation.
The MQS multiscale box plot is shown in the top row of Figure \ref{fig:sex}.
A common feature of ion channel data, called open channel noise, is that the noise variance in open states is often much higher than in closed states \citep[Section 3.4.4]{sakmann1995b}.
MQS is very robust against this heterogeneity while at the same time reliably detects most of gating events.
Figure \ref{fig:ic} shows that SMUCE \citep{frick2014}, R-FPOP \citep{fearnhead2017}, and WBS \citep{fryzlewicz2014} introduce a lot of artificial changes, because of increased variance for open channel noise.
QS \citep{eilers2005} (fifth row in Figure \ref{fig:ic}) misses most of the structural segment changes.
Although particularly tailored to this application, HSMUCE \citep{pein2017a} (and NOT(VAR) \citep{baranowski2019}), which assume heteroscedastic normal observations, do not seem to be superior to MQS here.
NOT(HT) shows a similar reconstruction to MQS(0.5).
NWBS misses considerable spikes that other methods can estimate.
Moreover, in contrast to the other methods, MQS explicitly quantifies the change in variance via the inter-quantile distance in the multiscale segment boxplot.
\section{Discussion}
In this work we proposed a new approach for quantile segmentation, that does not just provide an estimate with optimal detection rates, but also comes with honest confidence statements for the number of change-points, their locations, as well as confidence bands for the full underlying quantile function.
These results hold under the minimal assumption of independence between different observations and do not make any other distributional assumptions.
In simulations and real data examples we also observe empirically that the MQSE is very robust and consistently estimates changes in quantiles in various settings.
While other methods usually only work well in a particular setting (for which they are designed for), MQSE can be applied universally for any quantile segmentation task (under the independence assumption).
In the following we discuss a few possible extensions.

\paragraph{Multivariate change-point estimation}
One may wonder whether the MQS procedure can also be applied in a multivariate setting, i.e., for $Z_i\in\mathbb{R}^d$ for $d>1$.
However, it is well known that the univariate concept of quantiles as points of division of the mass of a probability distribution is not immediate to generalize for higher dimension. 
No single point in a multidimensional ($d>1$) space can divide the space in this way, what makes it difficult to extend the QSR-model (\ref{mod:model1}) in a straightforward manner to the multivariate setting (for a survey on multidimensional medians see \citep{small1990}). 
Some attempts have been made in terms of depth curves \citep{hallin2010}, and in this sense one could attempt to look for changes in these curves.
We believe this is an interesting venue to take, however outside of the scope of this paper. 
Note, that even for the more simple concept of component-wise quantiles, there is a severe computational challenge, since the inversion of the local tests leads to ellipsoids (instead of intervals as in the one-dimensional case), which are hard to intersect repetitively.

\paragraph{Change-point estimation for dependence structure}
The only assumption that we make for the theoretical analysis of the MQSE is independence of the sequential observations.
One might wonder whether this assumption can be further weakened.
Clearly, without any specific assumptions on the dependence structure there is no hope for consistent change-point estimation.
For the Gaussian case \cite{tecuapetla-gomez2017} provide some strategy to estimate the underlying covariance structure for piece-wise constant signals with m-dependent errors.
Taking such an estimated covariance structure, one might correct the local likelihood ratio tests for dependence (see \citep{frick2014} for the Gaussian mean case) and thus, preserve the confidence statements of MQS.
However, any consistent estimation of such covariance structure will unavoidably come at the cost of making additional assumptions on the underlying distribution.
The major strength of our approach is that such assumptions are not required (and instead we only rely on the independence assumption).
In practice, we still found that as long as the dependence structure is sufficiently weak, MQSE is robust to this.
Interestingly, a slightly negative autocorrelation can even be beneficial for the detection power of MQS, which is consistent with results on minimax detection boundaries in the autocorrelated case \cite{enikeeva2020}. 
We provide some simulation results in Appendix \ref{sec:ser_dep}.

\section*{Acknowledgments}
The authors acknowledge support of \textit{DFG-RTG 2088}, \textit{DFG-SFB 803 Z02}, and DFG Cluster of Excellence 2067 MBExC. MB was supported by DFG postdoctoral fellowship BE 6805/1-1. Helpful comments from Chris Holmes, Housen Li, and Florian Pein are gratefully acknowledged.

\begin{appendices}
\section{Proofs of Section \ref{sec:theory}}\label{sec:proofs}

The proofs of this section are similar in spirit to those in \citep{frick2014}. However, there are important diferences due to the discrete nature of the transformation (\ref{eq:transformtion}) and the corresponding convergence rates of the empirical quantiles. Before we prove the main results of Section \ref{sec:theory}, we require a couple of auxiliary results.

To this end, let $W_1,W_2,\ldots,W_n$ be i.i.d.\ Bernoulli distributed random variables with mean $\beta$ and let for $x\geq k\geq 0$
\[h_k(x)=x\log\frac{x}{k\beta}+(k-x)\log\frac{k-x}{k(1-\beta)}.\]
For fixed $n$ define the random variables
\[\xi(i,j)=\sqrt{2h_{j-i+1}\left(\sum_{k=i}^jW_k\right)}-\sqrt{2\log\left(\frac{\e \, n}{j-i+1}\right)}.\]
\begin{theorem}\label{theo:multiplecp}
Let $k\in\mathbb{N}$ with $k\geq 1$ and $q_n(\alpha)$ as in (\ref{eq:quantile}). Then
\[\mathbb{P}\left(\min_{1\leq s\leq k}\xi(i_s,j_s)> q_n(\alpha)\text{ for some }1\leq i_1\leq j_1<...< i_k\leq j_k\leq n\right)\leq \alpha^k.\]
\end{theorem}

\begin{proof}
Define the stopping times
\begin{align*}
\zeta_0(q)&=1,\\
\zeta_k(q)&= \min\left\{j>1:\max_{\zeta_{k-1}(q)< i\leq j}\xi(i,j)>q\right\}.
\end{align*}

Note that
\begin{align*}
&\zeta_{k+1}(q)-\zeta_k(q)+1= \min\left\{j-\zeta_k(q)>0:\max_{\zeta_{k}(q)< i\leq j}\xi(i,j)>q\right\}\\
&= \min\left\{j>1:\max_{\zeta_k(q)< i\leq j+\zeta_k(q)-1}\xi(i,j+\zeta_k(q)-1)>q\right\}\\
&= \min\left\{j>1:\max_{1< i\leq j}\sqrt{2h_{j-i+1}\left(\sum_{r=i}^{j}W_{r+\zeta_k(q)-1}\right)}-\sqrt{2\log\left(\frac{\e \, n}{j-i+1}\right)}>q\right\}.
\end{align*}
Consider the Markov process $\left(\sum_{i=1}^nW_i\right)_{n\in\mathbb{N}}$. 
By the strong Markov property, for any stopping time $\tau$ the process $(\sum_{i=1}^nW_{i+\tau})$ is independent of $W_1,\ldots,W_\tau$  conditioned on $\tau<\infty$. Note also that $(\sum_{i=1}^nW_{i+\tau})$ is identically distributed for all stoping times $\tau$, because it is the sum of $n$ i.i.d. random variables. This implies that 
\[\zeta_1(q)\,,\, \zeta_2(q)-\zeta_1(q)+1\,,\, \zeta_3(q) - \zeta_2(q) + 1 ,\, \zeta_4(q)-\zeta_3(q)+1\,, \ldots\]
are independent and identically distributed.
Therefore, for any $k\geq 1$ and $x>0$ it follows that
\[\mathbb{P}(\zeta_k(q)-1\leq x)=\mathbb{P}\left(\sum_{l=1}^k\zeta_l(q)-\zeta_{l-1}(q)\leq x\right)\leq \mathbb{P}(\zeta_1(q)-1\leq x)^k.\]
By definition $\zeta_1(q)\leq n$ implies that $M_n>q$, for $M_n$ as in (\ref{eq:mn}). Therefore, it follows from (\ref{eq:quantile}) that
\[\mathbb{P}(\zeta_1(q_n(\alpha))\leq n)\leq \mathbb{P}(M_n>q_n(\alpha))\leq \alpha.\]
\end{proof}

For the following two theorems let $Z_1,\ldots,Z_n$ be independent random variables with distribution functions $F_1,\ldots,F_n$ and $\beta\in(0,1)$ be given. Assume that the beta quantile is identical for all random variables, i.e.,
\[\theta_\beta \ZuWeis F_i^{-1}(\beta) \ZuWeis \inf\{x\in\mathbb{R}:F_i(x)\geq \beta\},\quad\mbox{for all }i=1,\ldots,n.\]
Define the empirical quantile and the minimal distribution function, the maximal distribution function, and the mean distribution function as:
\begin{align}
\begin{aligned}
\hat{\theta}_\beta &\ZuWeis \inf\left\{x\in\mathbb{R}: \sum_{i=1}^n \indE_{ \{Z_i \leq x\}} \geq n \beta\right\},\\
F_{\min}&\ZuWeis \min\{F_1,\ldots,F_n\},\\
F_{\max}&\ZuWeis \max\{F_1,\ldots,F_n\},\\
\overline{F}&\ZuWeis \frac{1}{n}\sum_{i=1}^nF_i.
\end{aligned}
\end{align}

Note that $F_\min$, $F_\max$, and $\overline{F}$ are again distribution functions with $\beta$-quantile $\theta_\beta$.

\begin{lemma}\label{lem:sumpbou}
Let $X_1,\ldots,X_n$ and $Y_1,\ldots,Y_n$ be independent Bernoulli random variables s.t. $X_i\sim B(p_i)$, $Y_i\sim B(p_{\min})$, and $Z_i\sim B(p_{\max})$ for all $i=1,\ldots,n$ and $p_{\min}=\min\{p_1,\ldots,p_n\}$, $p_{\max}=\max\{p_1,\ldots,p_n\}$. Then
\begin{align*}
\mathbb{P}\left(\sum_{i=1}^n X_i\leq c\right)&\leq\mathbb{P}\left(\sum_{i=1}^n Y_i\leq c\right),\\
\mathbb{P}\left(\sum_{i=1}^n X_i\leq c\right)&\geq\mathbb{P}\left(\sum_{i=1}^n Z_i\leq c\right) \quad\text{for all }c\in\mathbb{N}_0.
\end{align*}
\end{lemma}
\begin{proof}
Choose $W_i\sim B(p_\min/p_i)$ independently for all $i=1,\ldots, n$ and independently of $X_1,\ldots,X_n$. Let
\[\tilde{Y}_i=\begin{cases}
1&\text{if }X_i=1\text{ and }W_i=1\\
0&\text{else}
\end{cases}.\]
Then $\tilde{Y}_i\sim B(p_\min)$ and $\sum_{i=1}^n \tilde{Y}_i\leq \sum_{i=1}^n X_i$ a.s. Therefore,
\[\mathbb{P}\left(\sum_{i=1}^n X_i\leq c\right)\leq\mathbb{P}\left(\sum_{i=1}^n \tilde{Y}_i\leq c\right)=\mathbb{P}\left(\sum_{i=1}^n Y_i\leq c\right) \quad\text{for all }c\in\mathbb{N}.\]
The other inequality follows analogously.
\end{proof}

\begin{theorem}\label{theo:quantileaprox}
Let $\xi_{F,\beta}$ be as in Definition \ref{def:quantijump}. Then for any $\delta>0$
\begin{align*}
\mathbb{P}\left(\hat{\theta}_\beta-\theta_\beta>\delta\right)&\leq 2\exp\left(-2n\,\xi_{F_\min,\beta}(\delta)^2\right),\\
\mathbb{P}\left(\theta_\beta-\hat{\theta}_\beta>\delta\right)&\leq 2\exp\left(-2n\,\xi_{F_ \max,\beta}(-\delta)^2\right).
\end{align*}
\end{theorem}
\begin{proof}
\begin{align*}
\mathbb{P}\left(\hat{\theta}_\beta-\theta_\beta>\delta\right)&=\mathbb{P}\left(\inf\left\{x\in\mathbb{R}: \sum_{i=1}^n \indE_{\{ Z_i \leq x\}} \geq n \beta\right\}>\theta_\beta+\delta\right)\\
&=\mathbb{P}\left(\sum_{i=1}^n \indE_{ \{Z_i \leq \theta_\beta+\delta\} } < n \beta\right)
\end{align*}
Now, let $X_1,\ldots,X_n$ be i.i.d. random variables with distribution function $F_{\min}$. Then it holds
\[\mathbb{P}(Z_i<\theta_\beta+\delta)\leq\mathbb{P}(X_i<\theta_\beta+\delta)\quad\text{for all }i=1,\ldots,n.\]
From Lemma \ref{lem:sumpbou} it follows that
\begin{align*}
\mathbb{P}\left(\hat{\theta}_\beta-\theta_\beta>\delta\right)
&=\mathbb{P}\left(\sum_{i=1}^n \indE_{\{ Z_i \leq \theta_\beta+\delta\}} < n \beta\right)\\
&\leq\mathbb{P}\left(\sum_{i=1}^n \indE_{\{ X_i \leq \theta_\beta+\delta\}} < n \beta\right)\\
&=\mathbb{P}\left(F_{\min}(\theta_\beta+\delta)-\frac{1}{n}\sum_{i=1}^n \indE_{\{ X_i \leq \theta_\beta+\delta\}}>F_{min}(\theta_\beta+\delta)-\beta\right)\\
&\leq 2\exp\left(-2n\,\xi_{F_{\min},\beta}(\delta)^2\right),
\end{align*}
where the last inequality follows from the Dvoretzky-Kiefer-Wolfowitz inequality, see e.g.\ \citep{massart1990}. The other inequality follows analog.
\end{proof}

\begin{theorem}\label{theo:falsemultiscale}
Let $T_1^n$ be as in (\ref{eq:lrs}) and $W$ as in (\ref{eq:transformtion}). Then, for any $\delta,\epsilon > 0$
\[\mathbb{P}\left(T_1^n\left(W\left(Z, \theta_\beta+\delta\right), \beta\right)\leq \epsilon\right)\leq 2\exp\left\{-\left(\sqrt{2n}\,\xi_{\overline{F},\beta}(\delta)-\sqrt{\epsilon}\right)_+^2\right\}.\]
\end{theorem}
\begin{proof}
For any $p,q\in(0,1)$, Pinsker's inequality, see e.g.\ \citep[Lemma 2.5]{tsybakov2009}, implies
\begin{equation}\label{eq:pinsker}
    p\log\frac{p}{q}+(1-p)\log\frac{1-p}{1-q}\geq 2(p-q)^2.
\end{equation}
$W_1,\ldots,W_n$ are independent, but not identically, Bernoulli distributed random variables with mean $F_i(\theta_\beta+\delta)$. 
Define $\overline{W}=n^{-1}\sum_{i=1}^nW_i$, $\overline{F}(\theta_\beta+\delta)=n^{-1}\sum_{i=1}^nF_i(\theta_\beta+\delta)$, and assume that ${\xi_{\overline{F},\beta}(\delta)>\sqrt{\sfrac{\epsilon}{2n}}}$ (otherwise the assertion follows trivially). Then (\ref{eq:pinsker}) implies
\begin{align*}
\mathbb{P}\left(T_1^n\left(W\left(Z, \theta_\beta+\delta\right), \beta\right)\leq \epsilon\right)&\leq \mathbb{P}\left(2\left(\overline{W}-\beta\right)^2\leq \sfrac{\epsilon}{n}\right)\\
&=\mathbb{P}\left(\,\abs{\overline{W}-\beta}\leq \sqrt{\sfrac{\epsilon}{2n}}\right)\\
&=\mathbb{P}\left(\,\abs{\overline{W}-\overline{F}\left(\theta_\beta+\delta\right)-\beta+\overline{F}\left(\theta_\beta+\delta\right)}\leq \sqrt{\sfrac{\epsilon}{2n}}\right)\\
&\leq \mathbb{P}\left(\,\abs{\beta-\overline{F}\left(\theta_\beta+\delta\right)}-\abs{\overline{W}-\overline{F}\left(\theta_\beta+\delta\right)}\leq \sqrt{\sfrac{\epsilon}{2n}}\right)\\
&=\mathbb{P}\left(\,\abs{\overline{W}-\overline{F}\left(\theta_\beta+\delta\right)}\geq \xi_{\overline{F},\beta}(\delta)-\sqrt{\sfrac{\epsilon}{2n}}\right)\\
&\leq 2\exp\left(-2n\left(\xi_{\overline{F},\beta}(\delta)-\sqrt{\sfrac{\epsilon}{2n}}\right)^2\right)\\
&=2\exp\left\{-\left(\sqrt{2n}\,\xi_{\overline{F},\beta}(\delta)-\sqrt{\epsilon}\right)^2\right\},
\end{align*}
where the last inequality follows from Hoeffding's inequality.
\end{proof}

Let $\delta>0$ and note that $\xi_{\overline{F},\beta}(\delta)\geq  \xi_{F_{\min},\beta}(\delta)$, since $F_\min\leq\overline{F}$ implies
\begin{align*}
\overline{F}(\theta_\beta+\delta)-\beta\geq F_\min(\theta_\beta+\delta)-\beta\geq 0, 
\end{align*}
where the last inequality follows from the monotonictiy of $F_\min$ and $\delta>0$. Similarly, $\xi_{\overline{F},\beta}(-\delta)\geq  \xi_{F_{\max},\beta}(-\delta)$.

Now we are ready to prove the main results of Section \ref{sec:theory}.

\begin{proof}[Proof of Theorem \ref{theo:overestimation}]
Let $J(\vartheta)$ denote the set of change-points of $\vartheta$. First, note that
\begin{dmath*}
    {\mathbb{P}\left(\Sna> S+2s\right)=\mathbb{P}\left(T_n(Z,\hat{\vartheta})>q_n(\alpha),\,\forall \hat{\vartheta}\in \cpset \text{ with } \# J(\hat{\vartheta})\leq S+2s-1\right)}
    \leq {\mathbb{P}\left(T_n(Z,\hat{\vartheta})>q_n(\alpha),\,\forall \hat{\vartheta}\in \cpset \text{ with } J(\vartheta)\subseteq J(\hat{\vartheta}),\,\# J(\hat{\vartheta})\leq S+2s-1\right)}
    \leq {\mathbb{P}\left(T_n(Z-\vartheta,\hat{\vartheta}-\vartheta)>q_n(\alpha),\,\forall \hat{\vartheta}\in \cpset \text{ with } \# J(\hat{\vartheta}-\vartheta)\leq 2s\right)}
    \leq{\mathbb{P}\left(T_n(\tilde{Z})>q_n(\alpha),\,\forall \tilde{\vartheta}\in \cpset \text{ with } \#J(\tilde{\vartheta})\leq 2s )\right)
    =\mathbb{P}\left(\tilde{S}_{n,\alpha}>1+2s\right)},
\end{dmath*}
with $\tilde{Z}=Z-\vartheta$ and $\tilde{S}_{n,\alpha}$ as in (\ref{eq:hatK}) with $Z$ replaced by $\tilde{Z}$.
Thus, we can assume w.l.o.g.\ that $\vartheta\equiv\theta_0$ and $S=1$. 

Observe that $\Sna\geq 2s + 2$ implies that the multiscale constraint for the true regression function $\vartheta$ is violated on at least $s+1$ disjoint intervals $[\sfrac{i_1}{n},\sfrac{j_1}{n}],\ldots,$ $[\sfrac{i_{s+1}}{n},\sfrac{j_{s+1}}{n}]$ $\subseteq[0,1]$, that is
\[\sqrt{2T_{i_k}^{j_k}(W(Z,\vartheta_\beta),\beta)}-\sqrt{2\log\left(\frac{\e \, n}{j_k-i_k+1}\right)}> q(\alpha)\quad\text{ for all }\; 1\leq k\leq s+1\]
and it follows from Theorem \ref{theo:multiplecp} that
\begin{dmath*}
\mathbb{P}\left(\exists(1\leq i_1\leq j_1\leq...\leq j_{s+1}\leq n):\min_{1\leq k\leq {s+1}}\sqrt{2T_{i_k}^{j_k}(W,\beta)}-\sqrt{2\log\left(\frac{\e \, n}{j_k-i_k+1}\right)}\geq q_n(\alpha)\right)
\leq{\mathbb{P}\left(\exists(1\leq i_1\leq j_1\leq...\leq i_{s+1}\leq j_{s+1}\leq n):\min_{1\leq k\leq {s+1}}\xi(i_k,j_k)\geq q(\alpha)\right)}
\leq \alpha^{s+1}.
\end{dmath*}
\end{proof}

%

\begin{proof}[Proof of Theorem \ref{theo:underestimation}]
%
Define for $s=1,\ldots,S-1$ the intervals 
\[I_s=\left(\tau_{s}-\lambda_s/2,\tau_{s}+\lambda_s/2\right].\]

Note that these intervals are pairwise disjoint because $\lambda_s\leq \min\{\tau_s-\tau_{s-1},\tau_{s+1}-\tau_s\}$.

Let $\theta_s^+=\max\{\theta_s,\theta_{s+1}\}$ and $\theta_s^-=\min\{\theta_s,\theta_{s+1}\}$ and split the interval $I_s$ accordingly, i.e.
\[I_s^+=\{t\in I_s:\vartheta(t)=\theta_s^+\}\quad\text{ and }\quad I_s^-=\{t\in I_s:\vartheta(t)=\theta_s^-\}.\]
We are interested in the event that a function exists which is constant on $I_s$ and fulfills the multiscale constraints in both $I_s^-$ and $I_s^+$, i.e.
\begin{multline}\Omega_s=\left\{\exists\,\hat{\theta}\in\mathbb{R}:\sqrt{2T_{I_s^+}(W(Z,\hat{\theta}),\beta)}-\sqrt{2\log\frac{\e \, n}{\#I_s^+}}\leq q \text{ and}\right.\\\left.\sqrt{2T_{I_s^-}(W(Z,\hat{\theta}),\beta)}-\sqrt{2\log\frac{\e \, n}{\#I_s^-}}\leq q\right\}.\end{multline}
Observe that either $\hat{\theta}\leq\theta_s^+-\sfrac{\delta_s}{2}$ or $\hat{\theta}\geq\theta_s^-+\sfrac{\delta_s}{2}$ and define
\[\Omega_s^+=\left\{\exists\,\hat{\theta}\leq\theta_s^+-\delta_s/2:\sqrt{2T_{I_s^+}(W(Z,\hat{\theta}),\beta)}-\sqrt{2\log\frac{\e \, n}{\#I_s^+}}\leq q\right\},\]
\[\Omega_s^-=\left\{\exists\,\hat{\theta}\geq\theta_s^-+\delta_s/2:\sqrt{2T_{I_s^-}(W(Z,\hat{\theta}),\beta)}-\sqrt{2\log\frac{\e \, n}{\#I_s^-}}\leq q\right\}.\]
Due to the independence of $\Omega_s^+$ and $\Omega_s^-$ and the fact that $\Omega_s\subseteq \Omega_s^+\cup\,\Omega_s^-$, we get that
\[\mathbb{P}(\Omega_s)\leq 1-(1-\mathbb{P}(\Omega_s^+))(1-\mathbb{P}(\Omega_s^-)).\]
Next, we prove an upper bound for $\mathbb{P}(\Omega_s^-)$; a bound for $\mathbb{P}(\Omega_s^+)$ follows from symmetry.

Let $F_1,\ldots,F_{\abs{I_s^-}}$ be the distribution functions of the random variables in $I_s^-$, then $F_i^{-1}(\beta)=\theta_s^-$ for all $i=1,\ldots,\abs{I_s^-}$ . Moreover, let $F_{min}$ and $\overline{F}$ denote the minimum and mean distribution function of the random variables in $I_s^-$. Let 
\begin{equation}\zeta_{I_s^-}:=\inf\left\{x\in\mathbb{R}:\sum_{i \in I_s^-}\mathbbm{1}_{Z_i\leq x}\geq \abs{I_s^-}\beta\right\}\label{eq:empquant}
\end{equation}
be the empirical quantile of the observations on the interval $I_s^-$. Thus, for all $\hat{\theta}\geq \theta_s^-+\sfrac{\delta_s}{2}$, if $\zeta_{I_s^-}\leq \theta_s^-+\sfrac{\delta_s}{2}$, we get that 
\begin{align}\label{eq:proofHelp1}
\beta\leq \overline{W}\left(Z,\zeta_{I_s^-}\right)\leq \overline{W}\left(Z,\theta_s^-+\sfrac{\delta_s}{2}\right)\leq \overline{W}\left(Z,\hat{\theta}\right).
\end{align}
Moreover, the function 
\[f:x\mapsto \abs{I_s^-}\left(x\log\left(\frac{x}{\beta}\right)+(1-x)\log\left(\frac{1-x}{1-\beta}\right)\right)\]
is strictly convex with minimum $\beta$ and hence, $f\big|_{[\beta,1]}$ is strictly increasing. Thus, (\ref{eq:proofHelp1}) implies that for all $\hat{\theta}\geq \theta_s^-+\sfrac{\delta_s}{2}$, if $\zeta_{I_s^-}\leq \theta_s^-+\sfrac{\delta_s}{2}$
\[T_{I_s^-}\left(W\left(Z,\hat{\theta}\right),\beta\right)=f(W(Z,\hat{\theta}))\geq f\left(W\left(Z,\theta_s^-+\sfrac{\delta_s}{2}\right)\right)=T_{I_s^-}\left(W\left(Z,\theta_s^-+\sfrac{\delta_s}{2}\right),\beta\right)\]
and hence,
\begin{dmath*}
\mathbb{P}(\Omega_s^-)\leq {\mathbb{P}\left(\Omega_s^-\cap \left\{\zeta_{I_s^-}\leq \theta_s^-+\frac{\delta_s}{2}\right\}\right)+\mathbb{P}\left(\zeta_{I_s^-}>\theta_s^-+\frac{\delta_s}{2}\right)}
\leq {\mathbb{P}\left(T_{I_s^-}(W(Z, \theta_s^-+\sfrac{\delta_s}{2}), \beta)\leq \frac{\left(q+\sqrt{2\log(\sfrac{2\e}{\lambda_s})}\right)^2}{2}\right)+\mathbb{P}\left(\zeta_{I_s^-}>\theta_k^-+\frac{\delta_s}{2}\right)}
\leq {2\exp\left(-\frac{\left(\sqrt{2n\lambda_s}\,\xi_{\overline{F},\beta}(\sfrac{\delta_s}{2})-q-\sqrt{2\log(\sfrac{2\e}{\lambda_s})}\right)_+^2}{2}\right)}+{2\exp(-n\lambda_s\,\xi_{F_{\min},\beta}(\sfrac{\delta_s}{2})^2)},
\end{dmath*}
where the last inequality follows from Theorems \ref{theo:quantileaprox} and \ref{theo:falsemultiscale}. Moreover, note that $\xi_{\overline{F},\beta}(\sfrac{\delta_s}{2})\geq \xi_s(\sfrac{\delta_s}{2})$. Hence,
\[\mathbb{P}(\Omega_s)\leq 1-(1-\mathbb{P}(\Omega_s^+))(1-\mathbb{P}(\Omega_s^-))\leq1-\gamma_{n,s},\]
with $\gamma_{n,s}$ as in (\ref{eq:gamma}). For $s=1,\ldots,S-1$  define the random variables
\[X_s(\omega)=\begin{cases}
0& \text{if }\omega\in\Omega_s,\\
1&\text{otherwise}.
\end{cases}\]

Observe that $X_s=1$ implies that any function $\hat{\vartheta}\in \cpset$ with $T_n(W(Z,\vartheta),\beta)\leq q$ has at least one change-point on $I_s$. Since $I_1,\ldots,I_{S-1}$  are pairwise disjoint this implies $\hat{S}(q)-1\geq \sum_{s=1}^{S-1}X_s$ . Therefore 
\[\mathbb{P}\left(\hat{S}(q)\geq S\right)\geq \mathbb{P}\left(\sum_{s=1}^{S-1}X_s\geq S-1\right)=\prod_{s=1}^{S-1}(1-\mathbb{P}(\Omega_s))=\prod_{s=1}^{S-1}\gamma_{n,s}.\]
\end{proof}
\begin{proof}[Proof of Corollary \ref{cor:easyunder}]
Let $\gamma_{n,s}(q)$ be as in (\ref{eq:gamma}), $\Xi,\Lambda$ as in (\ref{eq:LambdaXi}), and $\gamma_n(q)$ be defined as $\gamma_n(q)=\min_{1\leq s\leq S-1}\gamma_{n,s}(q)$. Then
\begin{align}\begin{split}\gamma_n(q)&\geq \left[1-2\exp\left(-\frac{\left(2\sqrt{\sfrac{n\Lambda}{2}}\,\Xi-q-\sqrt{2\log(\sfrac{2\textup{e}}{\Lambda})}\right)_+^2}{2}\right)-2\exp(-n\Lambda\,\Xi^2)\right]^2\\
&=\left[1-2\exp\left(-\left(\sqrt{n\Lambda}\,\Xi-q/\sqrt{2}-\sqrt{\log(\sfrac{2\textup{e}}{\Lambda})}\right)_+^2\right)-2\exp(-n\Lambda\,\Xi^2)\right]^2.\end{split}
\label{eq:helpbound}
\end{align}
From Theorem \ref{theo:underestimation} it follows that
\[\mathbb{P}\left(\hat{S}(q)< S\right)\leq 1-\prod_{s=1}^{S-1}\gamma_{n,s}(q)\leq 1-\gamma_n(q)^{S-1}.\]
Using the inequality $(1-x)^m\geq 1-mx$ for all $x\in(0,1)$ and $m\in\mathbb{N}$ and (\ref{eq:helpbound}) it follows that
\begin{align*}\mathbb{P}\left(\hat{S}(q)< S\right)&\leq 4(S-1)\left[\exp\left(-\left(\sqrt{n\Lambda}\,\Xi- \sfrac{q}{\sqrt{2}} - \sqrt{\log(\sfrac{2\textup{e}}{\Lambda})}\right)_+^2\right)+\exp(-n\Lambda\,\Xi^2)\right]\\
&\leq 4(S-1)\textup{e}^{-n\Lambda\Xi^2}\left[\textup{e}^{2\sqrt{n\Lambda}\Xi\left(\sfrac{q}{\sqrt{2}}+\sqrt{\log(2\textup{e}/\Lambda)}\right)}+1\right],
\end{align*}
where the last inequality comes from the fact that $a^2-(a-b)^2\leq2ab$, with $a=\sqrt{n\Lambda}\,\Xi$ and $b=\sfrac{q}{\sqrt{2}}+\sqrt{\log(\sfrac{2\e}{\Lambda})}$. The result follows from the fact that $S\leq\Lambda^{-1}$
\end{proof}

\begin{proof}[Proof of Theorem \ref{theo:vanishingrates}]
From the proof of Corollary \ref{cor:easyunder} we get the following bound
\begin{dmath*}{\mathbb{P}\left(\hat{S}(q)< S\right)}\leq {4\Lambda_n^{-1}\left[\exp\left(-\left(\sqrt{n\Lambda_n}\,\Xi_n- \sfrac{q_n}{\sqrt{2}} - \sqrt{\log(\sfrac{2\textup{e}}{\Lambda_n})}\right)_+^2\right)+\exp(-n\Lambda_n\,\Xi_n^2)\right]}
={4\left[\exp({-\Gamma_{1,n}})+\exp({-\Gamma_{2,n}})\right]}
\end{dmath*}
with $\Gamma_{1,n}=\left(\sqrt{n\Lambda_n}\,\Xi_n- \sfrac{q_n}{\sqrt{2}} - \sqrt{\log(\sfrac{2\textup{e}}{\Lambda_n})}\right)_+^2+\log\Lambda_n\,$ and $\,\Gamma_{2,n}=n\Lambda_n\,\Xi_n^2+\log\Lambda_n$. Then a sufficient condition for $\mathbb{P}\left(\hat{S}(q)< S \right)\rightarrow 0$ is that $\Gamma_{1,n}\rightarrow\infty$ and $\Gamma_{2,n}\rightarrow\infty$ as $n\rightarrow\infty$. 

\textit{Cases 1 and 2:} If $\liminf\Lambda_n>0$, then a sufficient condition for $\Gamma_{1,n}\rightarrow\infty$ is that \[\sqrt{n\Lambda_n}\,\Xi_n-\sfrac{q_n}{\sqrt{2}}\rightarrow\infty.\]
This holds if $\sfrac{\sqrt{n}\Xi_n}{q_n}\rightarrow\infty$. Moreover, if this is the case then $\Gamma_{2,n}\rightarrow\infty$ as $n\rightarrow\infty$ and the proof is finished.

\textit{Case 3:} If $\Lambda_n\rightarrow 0$, assume that $\sqrt{n\Lambda_n}\Xi_n\geq(2+\epsilon_n)\sqrt{-\log\Lambda_n)}$ for a sequence $\epsilon_n$ such that $\epsilon_n\sqrt{-\log\Lambda_n}/q_n\rightarrow\infty$. Using the inequality $\sqrt{x+y}-\sqrt{x}\leq \sfrac{y}{(2\sqrt{x})}$ for $x,y\geq 0$ we obtain
\begin{align*}
\Gamma_{1,n}&\geq \left((2+\epsilon_n)\sqrt{-\log\Lambda_n}-\frac{q_n}{\sqrt{2}}-\sqrt{1+\log 2-\log\Lambda_n}\right)_+^2+\log(\Lambda_n)\\
&\geq \left((1+\epsilon_n)\sqrt{-\log\Lambda_n}-\frac{q_n}{\sqrt{2}}-\frac{1+\log2}{2\sqrt{-\log\Lambda_n}}\right)_+^2+\log(\Lambda_n)\\
&\geq \left(\epsilon_n\sqrt{-\log\Lambda_n}-\frac{q_n}{\sqrt{2}}-\frac{1+\log2}{2\sqrt{-\log\Lambda_n}}\right)_+^2
\end{align*}
where the last inequality comes from the fact that $(a+b)^2-a^2\geq b^2$, for $a,b\geq 0$. A sufficient condition for $\Gamma_{1,n}\rightarrow\infty$ is then that $\epsilon_n\sqrt{-\log\Lambda_n}/q_n\rightarrow\infty$, as it was assumed. Note that $\sqrt{n\Lambda_n}\Xi_n\geq(2+\epsilon_n)\sqrt{-\log\Lambda_n}$ implies 
\begin{align*}
    \Gamma_{2,n}&\geq \left[1-(2+\epsilon_n)^2\right]\log\Lambda_n\\
    &=(-3-4\epsilon_n-\epsilon_n^2)\log\Lambda_n
\end{align*}
which goes to infinity as by definition $\liminf_n\epsilon_n\geq 0$.
\end{proof}
%

\begin{proof}[Proof of Theorem \ref{theo:cprate}]
As in the proof of Theorem \ref{theo:underestimation}, define $S$ disjoint intervals 
\[I_s=(\tau_s-\epsilon_n,\tau_s+\epsilon_n)\subseteq[0,1),\]
and define $I_s^+$, $I_s^-$, $\theta_s^+$ and $\theta_s^-$ accordingly.
Now assume an estimator $\hat{\vartheta}$ of $\vartheta$, with estimated number of segments $\hat{S}$, such that $T_n(Z,\hat{\vartheta})\leq q$ and 
\[\max_{0\leq s\leq S-1}\min_{0\leq l\leq \hat{S}-1}|\hat{\tau}_l-\tau_s|>\epsilon_n.\]
In other words, there exists $0\leq s\leq S-1$ such that $|\hat{\tau}_l-\tau_s|>\epsilon_n$ for all $0\leq l\leq \hat{S}-1$, i.e. $\hat{\vartheta}$ does not have a change-point in the interval $I_s$. Then, as in the proof of Theorem \ref{theo:underestimation},
\begin{align*}
    &{\phantom{\leq}}\mathbb{P}\left(\exists\hat{\vartheta}\in \cpset:\, T_n(Z,\hat{\vartheta})\leq q \,\text{ and } \max_{0\leq s\leq S-1}\min_{0\leq l\leq \hat{S}-1}|\hat{\tau}_l-\tau_s|>\epsilon_n\right)\\
    \begin{split}\leq\mathbb{P}\left(\exists\hat{\theta}\in\mathbb{R}\text{ and some }s : \, \sqrt{2T_{I_s^+} (W(Z,\hat{\theta}),\beta)}-\sqrt{2\log\frac{\e \, n}{\epsilon_n}}\leq q \quad\text{ and }\right.\qquad\\\left.\sqrt{2T_{I_s^-}(W(Z,\hat{\theta}),\beta)}-\sqrt{2\log\frac{\e \, n}{\epsilon_n}}\leq q\right)\end{split}
\end{align*}
Finally, the assertion follows by replacing $\lambda_s$ by $\epsilon_n$ in the proof of Theorem \ref{theo:underestimation}.
\end{proof}
\section{Illustrative example with simulated data}\label{sec:example2}
\begin{figure}[th!]
\centering
\includegraphics[width=0.95\textwidth]{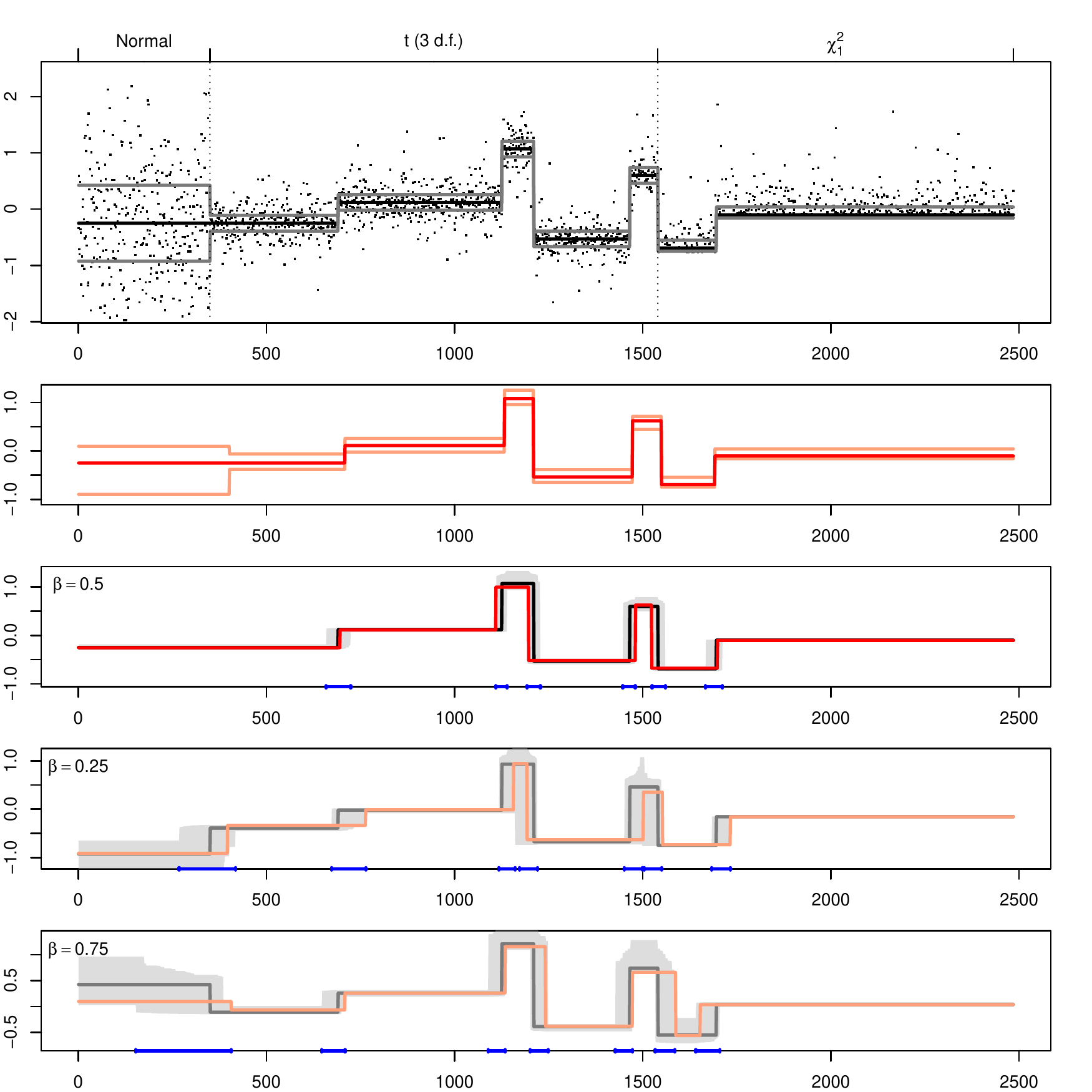}
\caption{\small First row: Observations $Z_1,\ldots,Z_n$, $n=2485$ as in the QSR-model with median function $\vartheta_{0.5}$ (black line) and 0.25- and 0.75-quantile functions $\vartheta_{0.25}, \vartheta_{0.75}$ (grey lines). Data comes from three different distributional regimes: Normal, $t_3$, and $\chi_1^2$.
Second row: The multiscale segment boxplot (MSB), with estimates (MQSE) for the median (red line), the 0.25- and 0.75-quantiles (light red lines), at nominal level $\alpha = 0.1$, see (\ref{eq:1statg}). 
Subsequent rows: True quantiles (black and gray solid lines) and MQSE (red and light red), together with  $90\%$ simultaneous confidence bands (gray area) and simultaneous confidence intervals for the change-point locations (blue intervals), for $\beta=0.5,0.75,0.25$, respectively. 
}
\label{fig:dif_reg}
\end{figure}

In the following we provide another example with synthetic simulated data which highlights MQS robustness to arbitrary distributional changes.
The data, shown as black dots in the top row of Figure \ref{fig:dif_reg}, comes from 3 different distributional regimes. 
The first 350 observations are drawn from a normal distribution with variance 1, the next 1190 observations from a t distribution with 3 d.f. and variance 0.05 and the last 945 from a $\chi^2$ distribution with 1 d.f. and variance 0.05. 
This data setting is covered by the QRS-model, which only assumes independence of the observations, but does not make any other distributional assumptions.
The black line in the top row of Figure \ref{fig:dif_reg} shows the true median of the underlying signal and the gray lines show the $0.25$ and $0.75$ quantiles.
For example, the median is a segment function $\vartheta_{0.5}\in\cpset$ (black line) with $S=7$ segments. 
Note that the QSR-model allows that different quantile curves may have entirely different segments, for example, the $0.25$- and $0.75$-quantile functions $\vartheta_{0.25}, \vartheta_{0.75}$ (gray lines) have an additional segment in the first regime ($S=8$).
The second row in Figure \ref{fig:dif_reg} shows the MQS estimates (MQSE) for the median (red line) and the 0.25- and 0.75-quantiles (light red lines) for the data as in the first row. 
The following three rows show the three estimates together with the corresponding confidence bands.
Note that, MQS accurately estimates change-point locations even when changes in the distribution occur, either at a change-point location or within a constant segment.
\section{MQSE with Wald-Wolfowitz runs statistic}\label{subsec:theoChoice}
In the following, we present an alternative approach to (\ref{eq:koenker}) for MQSE as a particular element in (\ref{eq:H}), which uses the transformed data in (\ref{eq:transformtion}) and is based on the Wald-Wolfowitz runs statistic \citep{wald1940}. 
Recall that, for $\vartheta = \vartheta_\beta$ the true underlying regression function in (\ref{eq:transformtion}), the transformed observations $W_1,\ldots,W_n$ are i.i.d.\ Bernoulli distributed with mean $\beta$.
Let $R$ be the number of runs of the sequence $W_1,\ldots,W_n$, i.e.,
\[R(W) \ZuWeis \# \{1 \leq i \leq n-1 \; : \; W_i \neq W_{i+1}\} + 1,\]
and $N_0(W) \ZuWeis \sum_{i=1}^n W_i$. Then, conditioned on $N_0 = n_0$ and $n - N_0 = n_1$, it holds that $R$ is asymptotically normally distributed with mean $\mu= {2n_1n_0}/{n}-1$ and variance $\sigma^2={2n_1n_0(2n_1n_0-n)}/({n^2(n-1)})$
see \citep{wald1940}.
Hence, for $\vartheta = \vartheta_\beta$ in (\ref{eq:transformtion}) it follows that $\Pp(R(W) = r, \; N_0(W) = k) = \Pp(R(W) = r, \;|\; N_0(W) = k)\Pp(N_0(W) = k)$ is asymptotically well approximated by
\begin{align}\begin{split} D(r,k)
\ZuWeis \binom{n}{k}\frac{\beta^k(1-\beta)^{n-k}}{\sqrt{2\pi \sigma^2}}\exp{\left(-\frac{(r-\mu(n,k))^2}{2\sigma(n,k)^2}\right)}.\end{split}\label{eq:testst}
\end{align}

\begin{definition}\label{def:mqse}
For a threshold $q > 0$, $\cH(q)$ as in (\ref{eq:H}), the MQS estimator (MQSE) is defined as
\begin{equation}\label{eq:mqse}
\hat{\vartheta}  \in  \argmax_{\vartheta\in\mathcal{H}(q)}D\left(R(W(Z,\vartheta)), N_0(W(Z,\vartheta))\right).
\end{equation}
\end{definition}

\begin{remark}\label{rem:uniqueness}
Note that the estimator MQSE may not be unique, as for two estimators $\hat{\vartheta},\tilde{\vartheta} \in \mathcal{H}(q)$ it is possible that $R(\hat{W})=R(\tilde{W})$ and $N_0(\hat{W})=N_0(\tilde{W})$, for $\hat{W}\ZuWeis W(Z,\hat{\vartheta})$ and $\tilde{W}\ZuWeis W(Z,\tilde{\vartheta})$ as in (\ref{eq:transformtion}).  For $n$ sufficiently large, we found it to be unique in most of the cases. When it is not unique we choose the first feasible change-point location.
\end{remark}

In the following, we provide a small simulation study to compare the segmentation performance of MQSE for the Koenker loss from equation (\ref{eq:koenker}) in the main text and the runs based loss from equation (\ref{eq:mqse}).
To this end, we considered a simple bump function with $n = 700$ equally spaced observations in an additive regression model for Gaussian and Cauchy noise, respectively.
The underlying median signal is $0$ for the first $300$ and last $300$ observations and $1$ elsewhere.
Table \ref{tab:comp} shows the MISE and the V-measure for $\beta = 0.25, 0.5, 0.75$.
For Gaussian noise, the estimate from Definition \ref{def:mqse} sightly outperforms the Koenker based estimate from equation (\ref{eq:koenker}) in the main text, in particular for the higher quantile $\beta = 0.75$.
For the heavy tailed Cauchy noise, it also outperforms the Koenker based estimate for the $0.75$ quantile, but not for the other two quantiles.
In summary, we find that both estimates perform comparable, with the estimate from Definition \ref{def:mqse} slightly outperforming the Koenker based estimate for higher quantiles.
Therefore, we conclude that this choice is not of major concern in practice.
Both estimates are available in the online implementation at \url{https://github.com/ljvanegas/mqs}.

\begin{table}[h]
\centering
\begin{tabular}{c|c|c|cc}
 Setting & $\beta$ & Estimator & MISE & V-m.\\
 \hline
\multirow{6}{*}{Normal} &\multirow{2}{*}{0.5} & \textbf{Runs}  & 0.09 &  0.74\\
  & &Koenker & 0.09 &  0.72\\
 \cline{2-5}
&\multirow{2}{*}{0.25} &Runs & 0.43 &  0.69\\
& &Koenker & 0.42 &  0.68\\
 \cline{2-5}
&\multirow{2}{*}{0.75} & \textbf{Runs}  & 0.43 &  0.72\\
& &Koenker & 0.49 &  0.66\\
 \hline
 \multirow{6}{*}{Cauchy} &\multirow{2}{*}{0.5} & Runs  & 0.03 &  0.90\\
  & & \textbf{Koenker} & 0.01 &  0.98\\
 \cline{2-5}
&\multirow{2}{*}{0.25} &Runs & 0.05 &  0.85\\
& & \textbf{Koenker} & 0.02 &  0.95\\
 \cline{2-5}
&\multirow{2}{*}{0.75} &\textbf{Runs}  & 0.05 &  0.84\\
& &Koenker & 0.08 &  0.78\\
 \hline
\end{tabular}
\caption{\label{tab:comp} \small MISE and V-measure of the two proposed estimators in equation (\ref{eq:koenker}) in the main text (Koenker) and Definition \ref{def:mqse} (Runs), respectively, for a simple bump function as described in the text, for additive Gaussian and Cauchy noise and quantiles $\beta = 0.25, 0.5, 0.75$. The outperforming method is marked bold in each case.
\smallskip
}
\end{table}
\section{Details of the Implementation}\label{sec:detailsImplementation}
Consider some given observations $Z_1,\ldots, Z_n \in \R$, quantile $\beta\in(0,1)$, and threshold $q\in\mathbb{R}$ (corresponding to some confidence level $\alpha$ as in (\ref{eq:mn})). 
In this section we give more details on the implementation of MQS.
Pseudocode for MQS is given in Algorithm \ref{estimator}, where the general dynamic programming structure is similar to the algorithm for SMUCE in \citep{frick2014}. 

\paragraph*{Bellman equation:}

%
%

Following a dynamic programming approach, for $j=1,\ldots,n$ we successively solve the estimation problem for $Z_1,\ldots,Z_j$ adding additional change-points when necessary.
When solving sub-problem $j$, that is, calculating the respective MQSE $\hat{\vartheta}_{1,j}$ for observations $Z_1, \ldots, Z_{j}$ we have already solved sub-problems $l$ with estimates $\hat{\vartheta}_{1,l}$ for $l = 1, \ldots, j-1$, where  $\hat{\vartheta}_{1,j-1}$ had $k$ change-points.
If we knew, that the last change-point of $\hat{\vartheta}_{1,j}$ was at location $i < j$, then 
\begin{equation}\label{eq:din_prog}
\hat{\vartheta}_{1,j}=\hat{\vartheta}_{1,i}\mathbbm{1}_{\{1,\ldots,i\}}+\hat{\theta}_{i+1,j}\mathbbm{1}_{\{i+1,\ldots,j\}}.
\end{equation}
Note that (\ref{eq:din_prog}) corresponds to a dynamic programming Bellman equation, which connects the solution of a previous sub-problem $i$ to the solution of the current sub-problem $j$.
For an illustration see Figure \ref{fig:implementation}, where the red line represents $\hat{\vartheta}_{1,i}$ and the blue line $\hat{\theta}_{i+1,j}$.
Thus, in order to find $\hat{\vartheta}_{1,j}$ we need to find its last change-point $i$.

\begin{figure}[th!]
\centering
\includegraphics[width=0.95\textwidth]{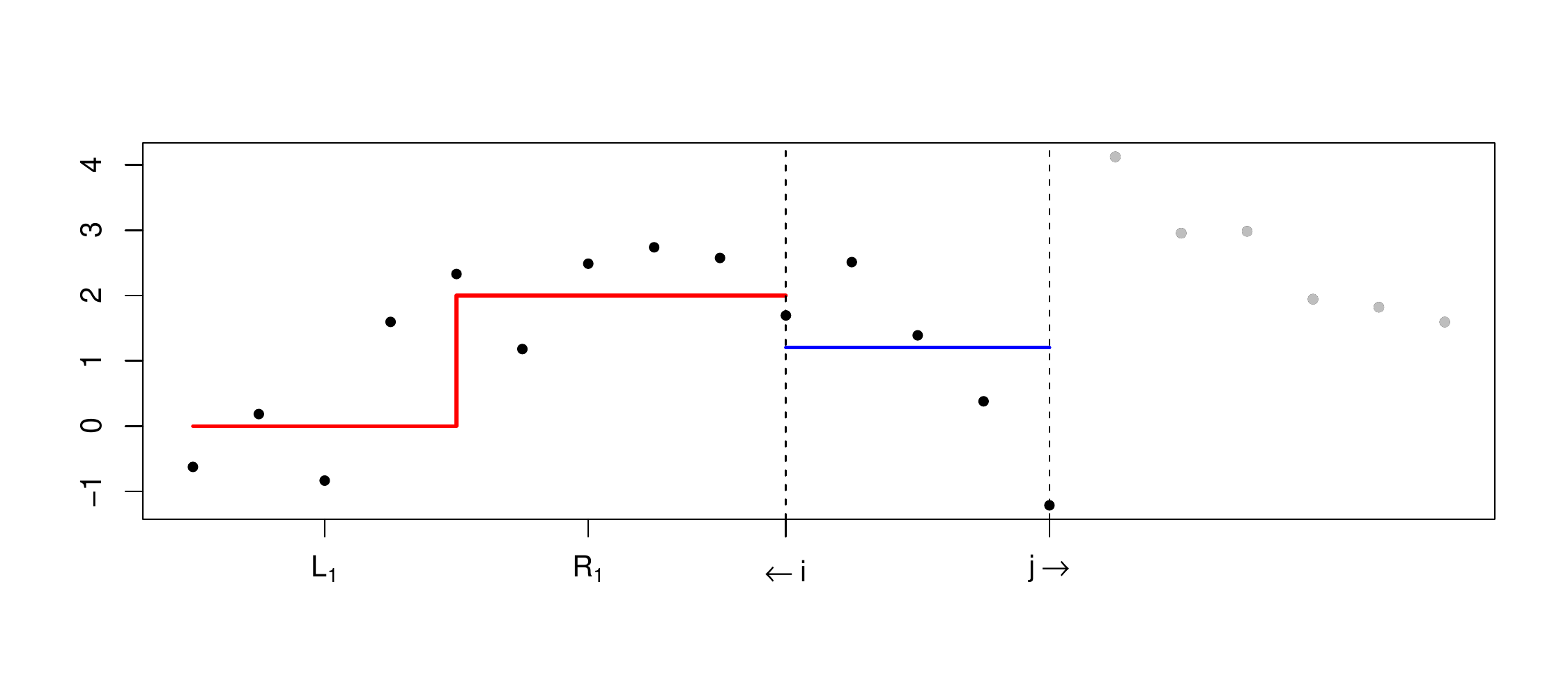}
\caption{\small Illustration of the dynamic programming scheme. Data $Z_1,\ldots,Z_n$ is given as dots. Black dots correspond to the first $j$ sub-problems that have already been solved. The remaining $n - j$ gray dots correspond to sub-problems that still need to be solved. The current candidate index for the last change-point of $\hat{\vartheta}_{1,j}$ is denoted by $i < j$. The algorithm has already calculated $\hat{\vartheta}_{1,i}$ (red line), and is currently calculating the last constant segment $\hat{\theta}_{i+1,j}$ (blue line). $L_1$ and $R_1$ represent the left and right bounds for the first change-point.
}
\label{fig:implementation}
\end{figure}

\paragraph*{Adding a change-point:} 


Recall that MQSE has a minimal number of change-points among all valid solutions in the set $\cH$ from equation  (\ref{eq:H}) (meaning that all local confidence statements are satisfied).
Therefore, when solving sub-problem $j$, in a first step, the algorithm checks whether there is a valid solution for $\hat{\vartheta}_{1,j}$ with $k$ change-points. 
Note that, as $\hat{\vartheta}_{1,j-1}$ has $k$ change-points, this implies that any valid solution for  $\hat{\vartheta}_{1,j}$ must have at least $k$ change-points, too.
Further, note that there is a valid solution for $\hat{\vartheta}_{1,j}$ with $k$ change-points if and only if there exists some $i < j$ with valid constant estimator on the interval $[i+1,j]$ such that $\hat{\vartheta}_{1,i}$ has $k - 1$ change-points.
When no such $i$ exists, this implies that $\hat{\vartheta}_{1,j}$ has exactly $k+1$ change-points.
Searching through all indexes $i$ where $\hat{\vartheta}_{1,i}$ has $k$ change-points, the algorithm finds the optimal location, in terms of the statistic $D$ in (\ref{eq:testst}), for $\hat{\vartheta}_{1,j}$'s last change-point.
The indexes $R_k$ (see pseudocode) keep track of those candidate locations that need to be considered when there are $k+1$ change-points, see next paragraph.

\paragraph*{Confidence intervals for change-points:}
While iterating over $j$ (and thereby dynamically increasing $k$), we keep track of indices $R_k$ and $L_k$.
$R_k\in\{1,\ldots,n\}$ is such that $\hat{\vartheta}_{1,R_k}$ has $k$ change-points and $\hat{\vartheta}_{1,R_k+1}$ has $k+1$ change-points, i.e., $R_k$ is the greatest index $j$ such that $\hat{\vartheta}_{1,j}$ has $k$ changes.
Similarly, $L_k\in\{1,\ldots,n\}$ is such that $\hat{\vartheta}_{L_k,R_k}$ has no change-point and $\hat{\vartheta}_{L_k+1,R_k}$ has one change-point, i.e., $L_k$ is the smallest index such that $\hat{\vartheta}_{L_k,R_k}$ is constant. 
Note that this implies that for any $\hat{\vartheta} \in \cH$ the location $\hat{\tau}_k$ of its $k$-th change-point is contained in the interval $[L_k,R_k]$. 
As the true underlying change-point function is contained in $ \cH(q)$ with probability at least $1-\alpha(q) + {o}(1)$ (recall Theorem \ref{theo:confidenceset}), this implies that $[L_k,R_k]$ is a valid confidence interval for its $k$th change-point location.

\paragraph*{Double-heap structures:}
In order to check whether a valid constant solution exists on a given interval, one has to intersection the confidence boxes from equation (\ref{eq:boxBij}) of the main text.
This can be done efficiently by updating the $\overline{m}_{1, j - i + 1}$ and $\underline{m}_{1, j - i + 1}$  quantiles for observations $Z_i, \ldots, Z_j$ from the respective quantiles for $Z_{i-1}, \ldots, Z_{j-1}$.
To implement this dynamic quantile update, we use double-heaps ($DH$) as described in \citep{astola1989}. 
Double-heaps are graphs that consist of two ordered binary trees, one decreasingly and one increasingly, with a common root (see Figure \ref{fig:doubleheap} for an example).
For a given quantile $\beta\in(0,1)$ and observations $Z_1,\ldots,Z_n$, whenever we have the observations stored in a double-heap with the right proportion of nodes stored below and above the root, we immediately get that the root is a $\beta$-quantile. 
For example, if we have the same number of data points stored above and below the root (as in Figure \ref{fig:doubleheap}), the root is always the median.
The benefit of working with double-heaps, is that they can be updated in $\log(n)$ time, since to maintain the ordering structure by adding (or replacing) one data point, one only has to go through a single branch.  

\begin{figure}
\centering
\begin {tikzpicture}[-latex ,auto ,node distance =1.5 cm and 1.5cm, on grid  ,
semithick ,
state/.style ={ circle ,top color =white ,
draw, text=blue , minimum width =1 cm}]
\node[state] (0) {4};
\node[state] (1) [above left= of 0] {8};
\node[state] (2) [above right =of 0] {6};
\node[state] (-1) [below left =of 0] {2};
\node[state] (-2) [below right =of 0] {-1};
\node[state] (3) [above left =of 1] {12};
\node[state] (4) [right = 2cm of 3] {9};
\node[state] (6) [above right =of 2] {11};
\node[state] (5) [left =2cm of 6] {15};
\node[state] (-3) [below left =of -1] {1};
\node[state] (-4) [right = 2cm of -3] {-3};
\node[state] (-6) [below right =of -2] {-2};
\node[state] (-5) [left =2cm of -6] {-5};

\path[style={sloped,anchor=south,auto=false}] (0) edge  node {$\geq$} (1);
\path (0) edge  node {}(2);
\path[style={sloped,anchor=south,auto=false}] (-1) edge node {$\leq$} (0);
\path (-2) edge node {} (0);
\path (1) edge node {} (3);
\path (1) edge node {} (4);
\path (2) edge node {} (5);
\path (2) edge node {} (6);
\path (-3) edge node{} (-1);
\path (-4) edge node{} (-1);
\path (-5) edge node{} (-2);
\path (-6) edge node{} (-2);

\draw[thick,dotted]     ($(6.north west)+(1,1)$) rectangle ($(3.south east)+(-1,-3.2)$);
\draw[red, thick,dotted] ($(-6.north west)+(1,3.2)$) rectangle ($(-3.south east)+(-1,-1)$);

\end{tikzpicture} 
\caption{Example of a double-heap for the median with 13 nodes.}
\label{fig:doubleheap}
\end{figure}
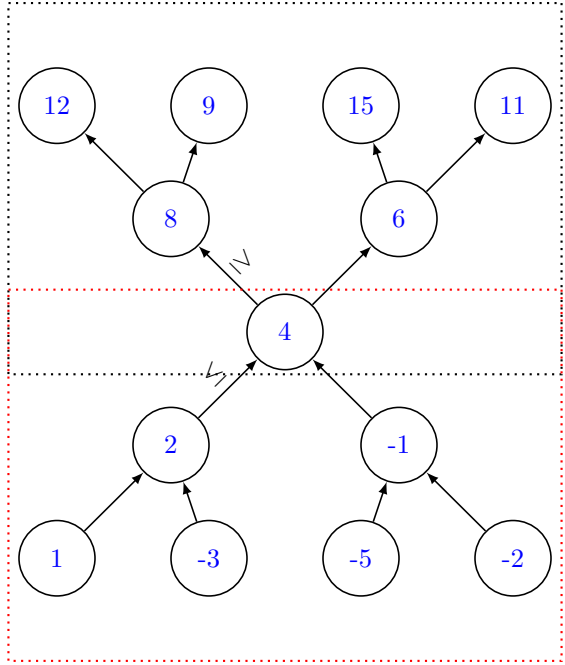

\paragraph*{Intersecting confidence boxes:} As detailed in Section \ref{sec:implementation} of the main text, the local tests can be inverted into intervals as in (\ref{eq:boxBij}).
Recall that their bounds correspond to the $\overline{m}_{j - i + 1}$ and $\underline{m}_{j - i + 1}$ quantiles of $Z_i, \ldots, Z_j$.
Therefore, we can store $2n$ double heaps (denotes as $\overline{DH}_{l}$ and $\underline{DH}_{l}$, $l = 1, \ldots, n$) for the upper and lower bounds of the intervals, updating them in each step by replacing one element.
In other words, to find the bounds for the interval $[i,j]$, we use the $(j-i+1)$ double-heap ($DH_{j-i+1}$), that has previously stored the bounds for $[i-1,j-1]$ and update it accordingly. 
Then, we can intersect this bounds recursively to check whether a constant valid solution exists.

\paragraph*{Updating the candidates and cost:} An additional double-heap structure $DH_\beta$ allows us to efficiently update the $\beta$-quantile for each interval $[i,j]$, and therefore, calculate the candidate function with $k$-th change-point $i$ (see Figure \ref{fig:implementation} blue line). 
If the $\beta$-quantile is inside the multiscale bound (i.e., if it is in the intersection of all local intervals) we choose it as $\hat{\theta}_{i+1,j}$, otherwise we choose either the upper or lower bound of the intersection.
The cost of (\ref{eq:din_prog}) for $i$ can also be calculated dynamically. 
The two quantities needed to calculate the statistic $D$ in (\ref{eq:testst}) are the sum of transformed data and the number of runs.
Both can be updated efficiently in this dynamic programming scheme.

\paragraph*{Confidence bands:}
For the constant segments $[R_k+1,L_{k+1}-1]$, the confidence bands are given by the intersection of all its local confidence boxes in (\ref{eq:boxBij}). 
For the confidence intervals $[R_k, L_k]$, the band is calculated as the union over all possible change-points $i$ of the corresponding intersected local boxes in (\ref{eq:boxBij}). 

\paragraph*{Computation time and space:}
The worst case computation time of this algorithm is $\mathcal{O}(n^2\log n)$, and this time is achieved with a constant estimator. 
In many cases the computation time is smaller, in particular when there are $O(n)$ jumps, in which case the computation time is of order $\mathcal{O}(n \log n)$. 
Compared to the SMUCE algorithm in \cite{frick2014} this is an increase of order $\log(n)$, which is due to the additional $\log(n)$ cost of quantile updates.
We also note that the required space of MQS is bigger than for SUMCE, since here we need to store double heap structures for each of the $2n + 1$ quantiles that need to be calculated. 
SMUCE needs $\mathcal{O}(n)$ space, while MQS requires $\mathcal{O}(n^2)$.

\begin{spacing}{0.8}
\begin{algorithm}
\footnotesize
\caption{Estimator MQS}\label{estimator}
\begin{algorithmic}[1]
\INPUT $Z_1,\ldots,Z_n$, $\beta\in(0,1)$ and $q\in\mathbb{R}$
\For{$j=1,\ldots,n$}
\State calculate $\underline{b}_{j}:=\underline{b}_{1,j}$ and $\overline{b}_{j}:=\overline{b}_{1,j}$ using (\ref{eq:boxBij}) in the main text\Comment{transformed data bounds}
\EndFor
\State $k=-1$
\State $R_0=1$
\State $R_1 = 1$
\State $j = 0$
\State Inititalize double-heaps $\underline{DH}_{\ell}$, $\overline{DH}_{\ell}$, $DH_{\beta}$, $l = 1, \ldots, n$ \Comment{see paragraph Double-heap structures}
\While{  $j < n$} 
\State $k=k+1$\Comment{add a change-point}
\For{$j=R_{k},\ldots, n$}
\For{$i=j,\ldots, R_{k-1}$}
\State replace $Z_{i-1}$ by $Z_j$ on $\underline{DH}_{j-i+1}$ 
\State replace $Z_{i-1}$ by $Z_j$ on $\overline{DH}_{j-i+1}$ \Comment{see paragraph Intersecting confidence boxes}
\State extract and intersect the bounds
\If{the intersection is not empty} \Comment{calculate candidate and cost for (\ref{eq:din_prog}) }
\State add $Z_i$ to $DH_\beta$ 
\State calculate $\hat{\theta}_{i,j}$
\State calculate $D$ \Comment{see paragraph Updating the candidates and cost}
\If{$D_i<D_{1,j}$}\Comment{new $D$ is lower$\rightarrow$ update optimal $i$}
\State $\hat{\vartheta}_{1,j}=\hat{\vartheta}_{1,i-1}\mathbbm{1}_{\{1,\ldots,i-1\}}+\hat{\theta}_{i,j}\mathbbm{1}_{\{i,\ldots,j\}}$
\EndIf
\Else
\State \textbf{break}\Comment{add a change-point}
\EndIf
\EndFor
\If{$i=R_k$}
\State$L_k=j+1$\Comment{update the left bound}
\EndIf
\EndFor
\State $R_{k+1}=i$ \Comment{update the right bound}
\State $L_{k+1}=R_k$
\EndWhile
\RETURN $\hat{\vartheta}$
\end{algorithmic}
\end{algorithm}
\end{spacing}
\subsection{The Multiscale Segment Boxplot}\label{sec:mbp}
When calculating the Multiscale Segment Boxplot (MSB)  (0.25-, 0.5-, and 0.75-quantiles) and two or more of the quantiles have an estimated change-point at a similar location, we might prefer an estimator that chooses the same location, instead of slightly different locations.
This will also avoid, with high probability, intersections in the three different quantiles  which eases interpretation (see e.g. \citep{he1997, chernozhukov2010}). 
Therefore, for the MSB we will modify the estimated MQSE change-point locations whenever the changes in different quantiles are sufficiently close.
Thereby, we ensure that the resulting estimator is still an element of the set $\cH$ in (\ref{eq:H}) for each quantile, thus all theoretical guarantees given in Section \ref{sec:theory} in the main text still hold.
Essentially, whenever the confidence intervals for change-points of two quantiles overlap, we select a common change-point for the MSB.

To do this, we first run the MQS algorithm to obtain the confidence statements, recall Section \ref{sec:detailsImplementation}.
To ensure locally simultaneous changes, if the confidence intervals of two or three quantiles intersect, we choose a common change-point inside the intersection.
The particular choice we make is the mean of the corresponding MQS change-point estimators if it is inside the intersection, otherwise we choose the upper or lower limit of the interval.
Then, after choosing the change-point, the estimated signal value is the empirical quantile of the delimited segment whenever this is contained in the confidence band. 
Otherwise we choose the upper or lower limit of the band.

By restricting to locally simultaneous changes, one typically also avoids the problem of quantile-crossing, that is, the estimated quantile function for some quantile $\beta$ exceeds in some region the estimate for some quantile $\beta^\prime > \beta$.
More generally, when estimating $k$ quantile functions simultaneously, using a Bonferroni correction, we obtain simultaneous coverage of the underlying quantile functions by our confidence bands with probability at least $1 - k \alpha$.
As the underlying quantile functions do not cross, by making an appropriate choice within the confidence band, quantile crossing happens with probability at most $k \alpha$.
In practice, we found that by making these adjustment quantile crossing for the MSB has a very low probability of happening. 
In 1,000 Monte Carlo simulations of MSB with $\alpha = 0.05$ for the Gaussian additive model as in Figure \ref{fig:bigex}, we did not find any quantile crossing.
\section{Simulation results on confidence statements}\label{sec:conf}

MQS also provides confidence statements for the number of segments, the change-point locations, and finally the underlying segment function itself (see Theorem \ref{theo:confidenceset}). 
In this section, we investigate the behavior of such confidence statements in the additive error models as considered in Section \ref{subsec:simplesec} in the main text, see Figure \ref{fig:bigex}, with $n=2485$. 

Simulation results are shown in Table \ref{tab:confidencetable}. 
As demonstrated in Section \ref{subsec:simplesec} in the main text, the number of segments is estimated correctly with high probability for $\beta=0.5,0.75$ and slightly underestimated for $\beta=0.25$. Our method gives the theoretical guarantee that $P(\hat{S} \leq S) \geq 1 - \alpha$ (see (\ref{eq:1statg})). In Table \ref{tab:confidencetable}, we give the frequency of $\hat{S}\leq S$. For all considered examples, (\ref{eq:1statg}) is fulfilled for a much larger level than required. 
Moreover, given that the number of segments is estimated correctly, we construct simultaneous confidence intervals for the true change-point locations. 
In Table \ref{tab:confidencetable} column "CI cov." (short for Confidence Interval coverage), we give the frequency that all the change-points $\tau_s$ were inside the respective confidence intervals (given that $\hat{S}=S$). We find that with very high probability, all the change-points are covered by the confidence intervals. 
We also provide confidence bands that cover the true function with probability at least $1-\alpha$ (see Theorem \ref{theo:confidenceset}), given that the number of change-points is estimated correctly. In Table \ref{tab:confidencetable} column "CB cov." (short for Confidence Band coverage), we give the frequency that the true function $\vartheta$ was fully inside the confidence band (given that $\hat{S}=S$). Here, we see that the nominal level is kept, but the confidence bands are, in general, too conservative.
%
%
\begin{table}
\centering
\resizebox{\textwidth}{!}{%
\begin{tabular}{ c|c|c|c|c|c|c|c|c|c } 
& $\beta$& $1-\alpha$& $\mathbb{P}(\hat{S}\leq S)$ & CI cov. &CB cov.&&$\mathbb{P}(\hat{S}\leq S)$ & CI cov. &CB cov.\\
\hline
\multirow{9}{*}{\shortstack{Normal\\$\sigma^2=0.04$}} &\multirow{3}{*} {0.5}& 0.99& 1.000 & 1.000 & 0.998&\multirow{9}{*}{\shortstack{t (3 d.f.)\\$\sigma^2=0.04$}}&1.000 & 1.000 & 0.996\\
& & 0.95& 1.000 & 1.000 & 0.977& &1.000 & 1.000 & 0.974\\
& & 0.90& 0.997 & 1.000 & 0.936& & 0.998 & 1.000 & 0.947\\
\cline{2-6}
\cline{8-10}
&\multirow{3}{*} {0.25}& 0.99& 1.000 & 1.000 & 1.000& &1.000 & 1.000 & 1.000\\
& & 0.95& 1.000 & 1.000 & 0.994& &1.000 & 1.000 & 0.996\\
& & 0.90& 0.999 & 1.000 & 0.951& &0.998 & 1.000 & 0.948\\
\cline{2-6}
\cline{8-10}
&\multirow{3}{*} {0.75}& 0.99& 1.000 & 1.000 & 0.998& &1.000 & 1.000 & 0.998\\
& & 0.95& 0.999 & 1.000 & 0.962& & 0.999 & 1.000 & 0.972\\
& & 0.90& 0.998 & 1.000 & 0.953& & 0.998 & 0.999 & 0.954\\
\hline
\multirow{9}{*}{Cauchy} &\multirow{3}{*} {0.5}& 0.99& 1.000 & 1.000 & 0.998&\multirow{9}{*}{\shortstack{$\chi^2_3$\\$\sigma^2=0.04$}}&1.000 & 1.000 & 0.993\\
& & 0.95& 0.999 & 1.000 & 0.965& &0.999 & 1.000 & 0.966\\
& & 0.90& 0.995 & 0.998 & 0.928& &1.000 & 0.998 & 0.936\\
\cline{2-6}
\cline{8-10}
&\multirow{3}{*} {0.25}& 0.99& 1.000 & 1.000 & 1.000& &1.000 & 1.000 & 1.000\\
& & 0.95& 1.000 & 1.000 & 0.997& &1.000 & 1.000 & 0.994\\
& & 0.90& 1.000 & 1.000 & 0.943& &0.998 & 1.000 & 0.955\\
\cline{2-6}
\cline{8-10}
&\multirow{3}{*} {0.75}& 0.99& 1.000 & 1.000 & 0.995& &1.000 & 1.000 & 0.992\\
& & 0.95& 1.000 & 1.000 & 0.971& &1.000 & 1.000 & 0.977\\
& & 0.90& 0.999 & 0.999 & 0.932& &0.998 & 0.999 & 0.952\\ 

\end{tabular}
}
\caption{\label{tab:confidencetable}Probability of estimating correctly the number of segments, frequency of confidence interval  coverage of the true c.p's (CI cov.), and frecuency of confidence band coverage of the true function (CB cov.), for data as in Figure \ref{fig:bigex} with different levels $\alpha$ and $n=2485$. 
}
\end{table}
\section{MQS under serial dependence}\label{sec:ser_dep}
The only assumption that is required for our theoretical results to be valid is independence of the observations as specified in the QSR model (\ref{eq:model1}).
In this section we explore empirically the MQS performance when this assumption is violated.
To this end, we consider an additive regression model with autoregressive noise.
More specifically, for $n = 500$ we consider observations of the form $Y_i = \mu_i+\sqrt{1-\theta^2}X_i$, where $\mu$ is a simple bump function with its three segments at $[1,125]$, $[126,375]$ and $[376,500]$ taking values $0$, $1$, $0$, respectively, see Figure \ref{fig:ar_ex}.
For $X_i$ we consider an AR(1) process, that is $X_{i+1} = \theta X_i +  \epsilon_t$ with $\epsilon_t\sim\cN(0,1)$, for different values of $\theta$ with $|\theta|<1$.
The autocorrelation function of the process $X_i$ is then given by $\rho_h = \theta^h$, where $h$ represents the lag.
Thereby, positive values of $\theta$ correspond to a positive autocorrelation, negative values of $\theta$ to negative autocorrelation, and $\theta = 0 $ corresponds to the independent case.
Note that since $X_i$ has variance $1/(1-\theta^2)$ we multiply the error by the variance stabilizing scaling $\sqrt{1-\theta^2}$.

Figure \ref{fig:ar_ex} shows examples for three different values of $\theta$. 
Figure \ref{fig:ar_comp} illustrates estimation performance as the parameter $\theta$ deviates from 0 for the example described above. 
For small $\theta$ values (i.e., small auto-correlation), MQS performs comparably to the independent case (red line) in terms of number of segments, MISE, and V-measure. 
We highlight that for negative $\theta$ the probability of estimating the correct number of segments can be even higher than in the independent case (see upper-left corner of Figure \ref{fig:ar_comp}).
This is consistent with recent general results on minimax detection boundaries for bump detection with stationary Gaussian error \cite{enikeeva2020}.
When $\theta$ gets very large the strong postive autocorrelation results in overestimation of the number of segments for the MQS (see upper-right corner of Figure \ref{fig:ar_comp}). 

\begin{figure}[th!]
\centering
\includegraphics[width=0.95\textwidth]{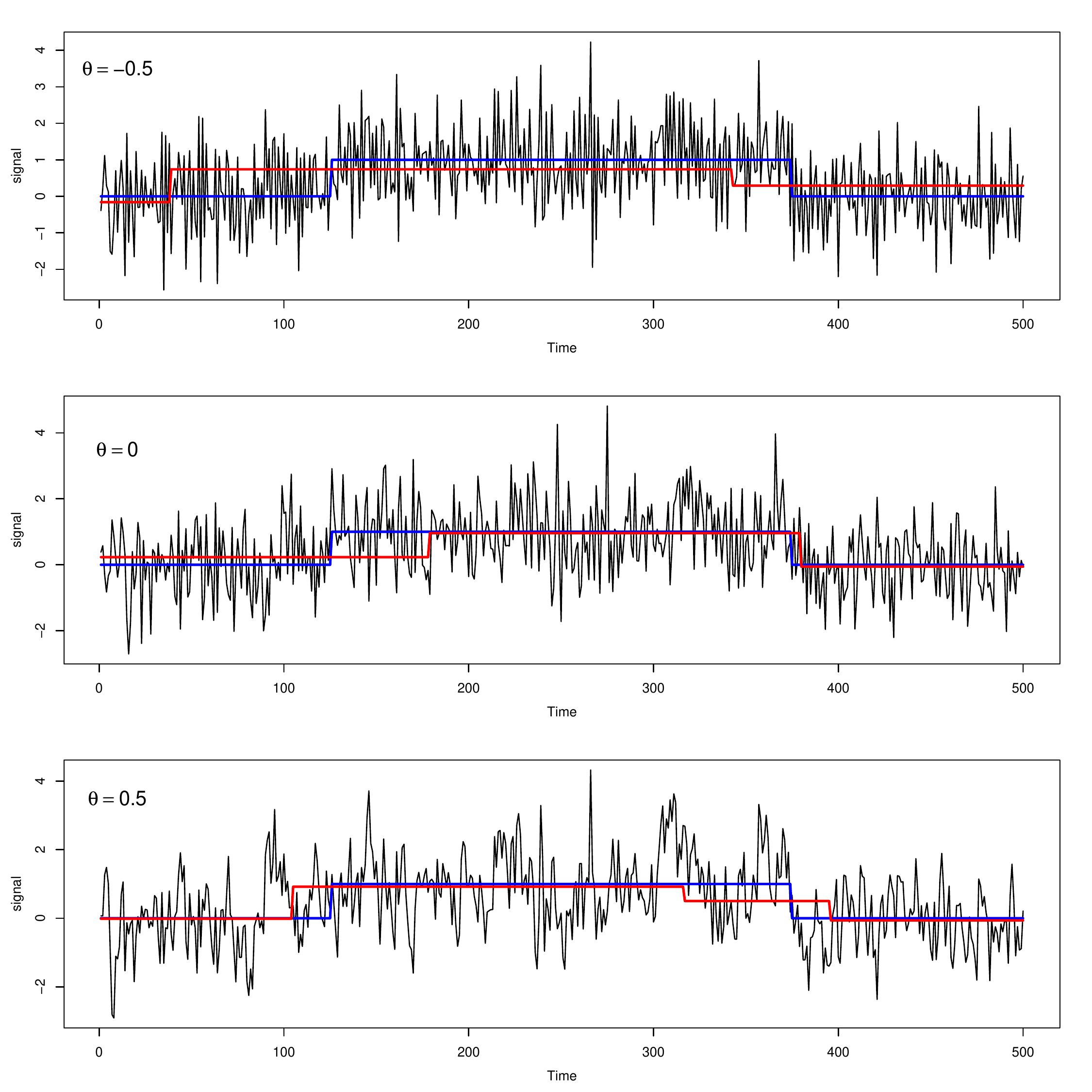}
\caption{\small 500 data points from a simple additive model $Y_i = \mu_i+\sqrt{1-\theta^2}X_i$, where $\mu_i$ is a simple bump function (blue line), and $X_i$ follows an AR(1) model with parameter $\theta$, together with MQSE (red line) for $\alpha = 0.05$.
}
\label{fig:ar_ex}
\end{figure}

\begin{figure}[th!]
\centering
\includegraphics[width=0.95\textwidth]{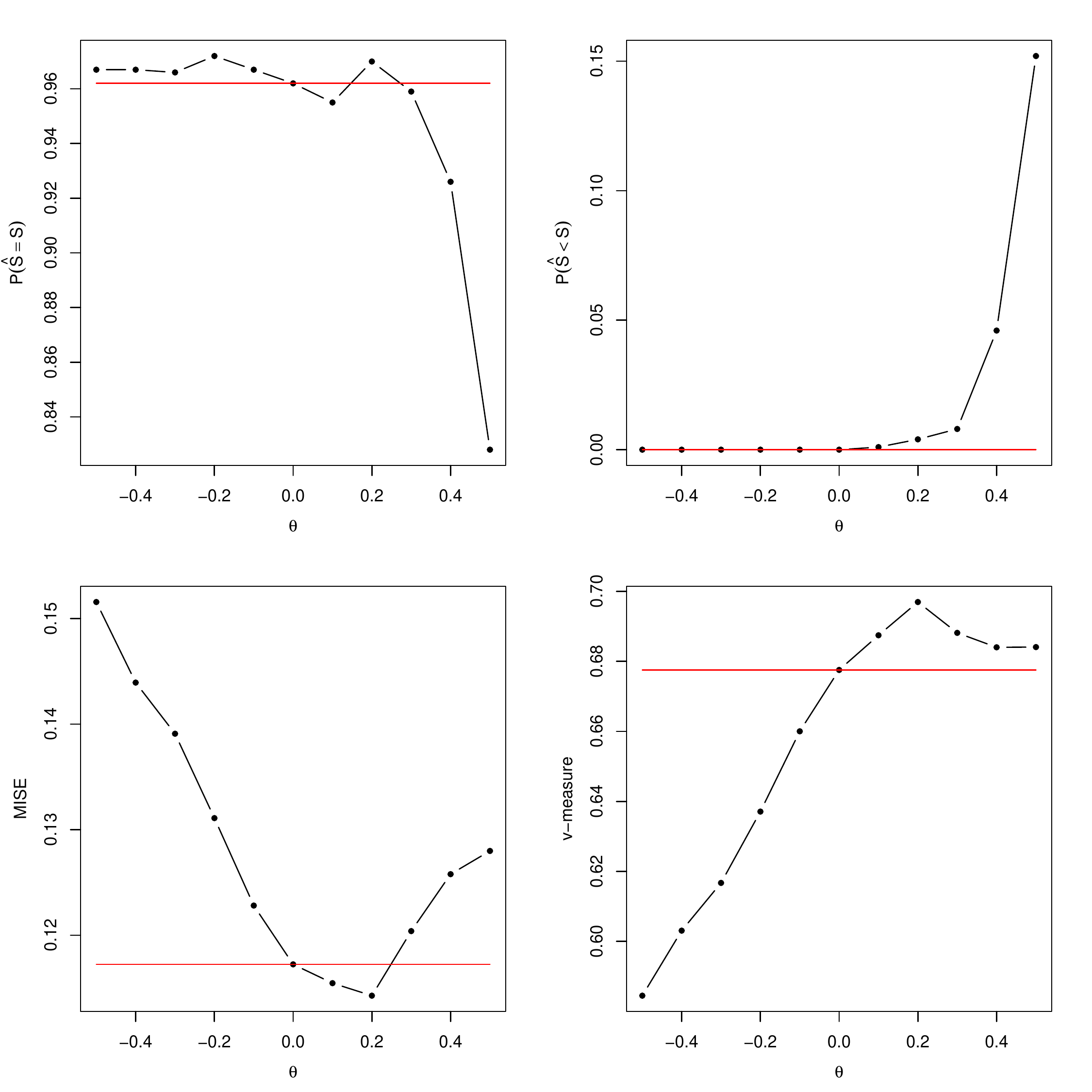}
\caption{\small Comparison between different values of the parameter $\theta$ as in Figure \ref{fig:ar_ex}. All values where obtained from 1,000 Monte-Carlo runs. The upper plots show changes in $P(\hat{S}=S)$ and $P(\hat{S}>S)$, respectively, where $\hat{S}$ is MQS' estimated number of segments. The lower plots show the MISE and V-measure for the MQSE for $\alpha = 0.05$.
}
\label{fig:ar_comp}
\end{figure}
\section{Additional Tables and Figures}

\begin{table}
\centering
\small
\resizebox{\textwidth}{!}{%
\begin{tabular}{ ccccccccc } 
&MQS&SMUCE& HSMUCE&WBS&R-FPOP&QS&NOT&NWBS\\
\shortstack{Confidence\\statements}&\checkmark & \checkmark & \checkmark & $\times$ & $\times$& $\times$&$\times$&$\times$\\
\shortstack{Consistency results}&\checkmark & \checkmark & \checkmark & \checkmark & \checkmark & $\times$&\checkmark & \checkmark\\
\shortstack{No distributional\\ assumptions}&\checkmark & $\times$ & $\times$ & $\times$ & \checkmark & \checkmark &$\times$& \checkmark\\
Target a specific quantile&\checkmark & $\times$ & $\times$ & $\times$& $(\times)$ & \checkmark & $\times$ & ($\times$)\\
\shortstack{Robust to outliers} &\checkmark & $\times$ & \checkmark & $\times$ & \checkmark & \checkmark & \checkmark & \checkmark \\
\shortstack{Robust to heterogeneity} &\checkmark & $\times$ & \checkmark & \checkmark & $\times$&\checkmark & \checkmark & $\times$\\
\shortstack{Computation time [s] \\(see caption)}&0.65&0.01 &0.02&0.07&0.01 &38&0.35&63.8
\end{tabular}
}
\caption{\label{tab:summarytable} \small Summary table of the characteristics of MQS versus SMUCE \citep{frick2014}, HSMUCE \citep{pein2017a}, WBS \citep{fryzlewicz2014}, R-FPOP \citep{fearnhead2017}, QS \citep{eilers2005}, NOT (HT contrast) \citep{baranowski2019}, and NWBS \citep{padilla2019}. The computation time comparison is based on normal noise ($n=1988$), mean as in Figure \ref{fig:bigex} and variance 0.04. With "robust to heterogeneity" we mean that the median/mean estimator is robust again changes in variance.
}
\end{table}

\begin{figure}[b!]
\centering
\includegraphics[width=0.9\textwidth]{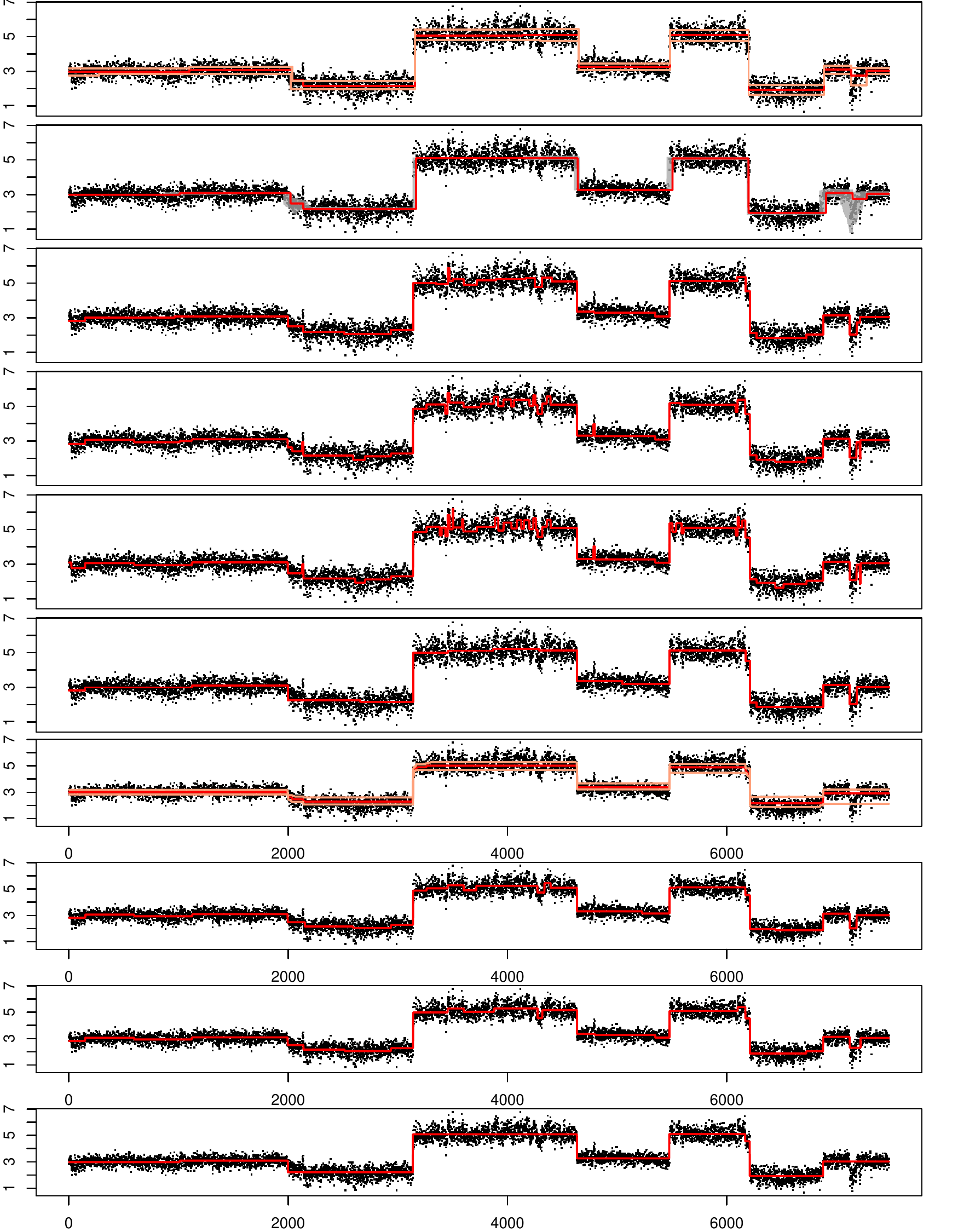}
  \caption{  Preprocessed WGS data (black dots) of cell line LS411 from colorectal cancer and different estimators for the underlying CNA's (red lines). Sequencing was performed by Complete Genomics in collaboration with the Wellcome Centre for Human Genetics at the University of Oxford.
From top to bottom: MSB with $\alpha = 0.01$, MQSE for the median with confidence band for $\alpha = 0.01$, SMUCE \citep{frick2014}, WBS \citep{fryzlewicz2014}, R-FPOP \citep{fearnhead2017}, HSMUCE \citep{pein2017a}, QS \citep{eilers2005} (estimated $0.25$, $0.5$, and $0.75$ quantiles), NOT (HT) and (VAR) \citep{baranowski2019}, and NWBS \citep{padilla2019}.}
\label{fig:wgsmulex}
\end{figure}

%
\begin{figure}
\centering
\includegraphics[width=\textwidth]{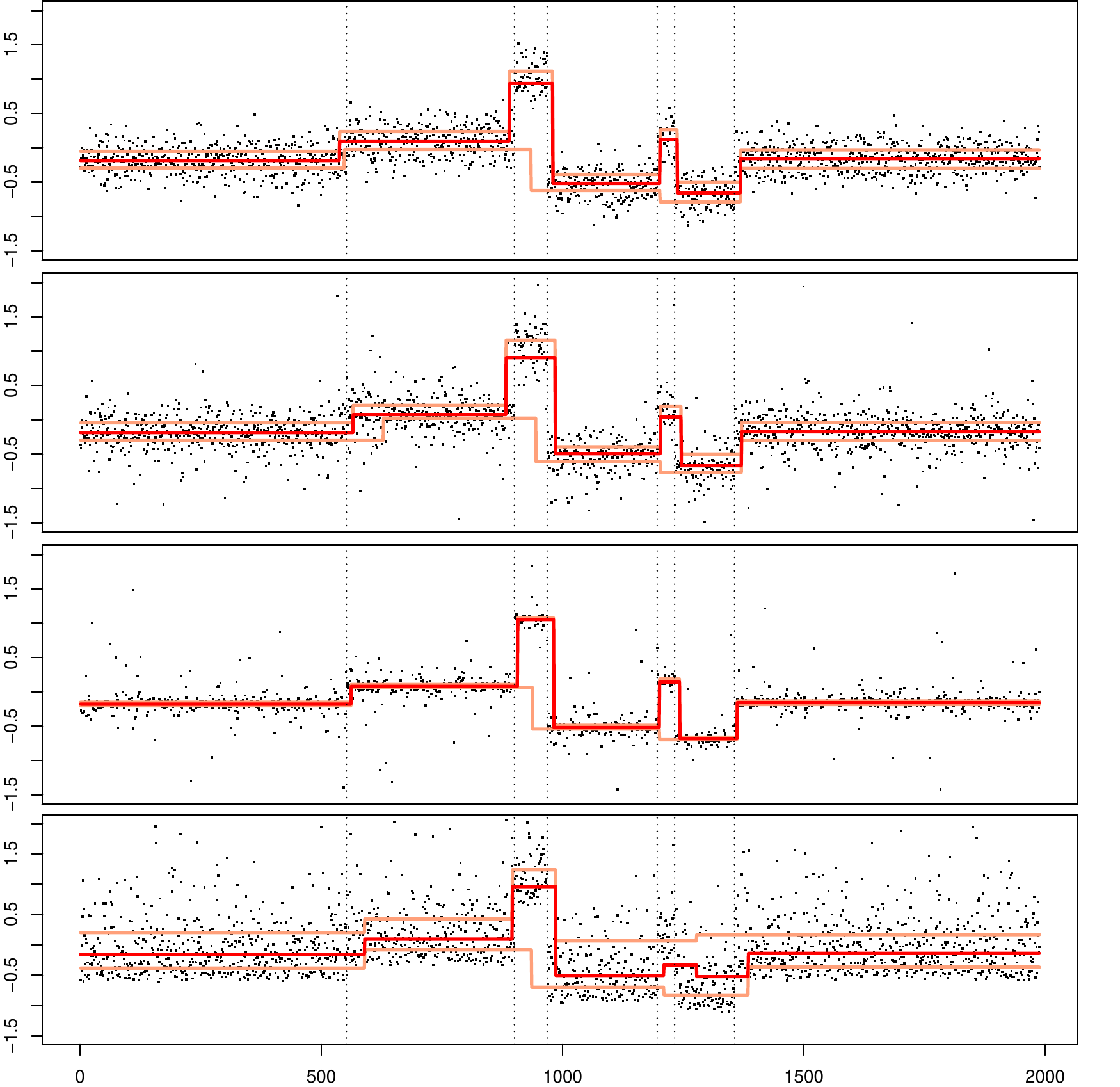}
\caption{Data as in (\ref{eq:simple}) (black dots) for $n = 1,988$ and different error terms. The true underlying segment changes are shown as dotted vertical lines.
The corresponding MSB with threshold parameter $\alpha = 0.1$ is shown in red (median) and light red lines (0.25- and 0.75 quantiles).
From top to bottom: normally distributed with variance $\sigma^2 = 0.04$, $t$ distributed with $3$ degrees of freedom and variance $\sigma^2 = 0.04$, Cauchy distributed, $\chi^2$ distributed with $3$ degrees of freedom and variance $\sigma^2 = 0.04$.
}
\label{fig:bigex}
\end{figure}

\begin{table}
\centering
\resizebox{0.92\textwidth}{!}{%
\begin{tabular}{ c|c|cc >{\columncolor[gray]{0.8}}c cc|c|c|c|c|c| >{\columncolor[gray]{0.8}}c cc|c|c } 
 &Method&$\leq 5$&6&$7$&$8$&$\geq 9$&MIAE&V-m.& &$\leq 5$&6&$7$&$8$&$\geq 9$&MIAE&V-m.\\
 \hline
 \multirow{6}{*}{\shortstack{Normal\\$\sigma^2=0.04$}}&MQS(0.5)& 0.00 & 0.70 & 99.20 & 0.10 & 0.00 & 3.83 & 9.29& \multirow{6}{*}{\shortstack{t (3 d.f.)\\$\sigma^2=0.04$}}& 0.00 & 10.30 & 89.50 & 0.20 & 0.00 & 3.38 & 9.33\\
 &SMUCE & 0.00 & 0.00 & 99.70 & 0.30 & 0.00 & 0.90 & 9.96& & 0.00 & 0.00 & 0.00 & 0.00 & 100.00 & 2.41 & 7.72\\
  &HSMUCE & 0.00 & 0.00 & 98.40 & 1.60 & 0.00 & 0.91 & 9.95& & 0.00 & 0.00 & 99.50 & 0.50 & 0.00 & 0.91 & 9.96\\
 &R-FPOP& 0.00 & 0.00 & 96.90 & 3.10 & 0.00 & 0.92 & 9.95& & 0.00 & 0.00 & 95.30 & 4.20 & 0.50 & 0.72 & 9.97\\
 &WBS& 0.00 & 0.00 & 98.00 & 1.90 & 0.10 & 0.97 & 9.95&  & 0.00 & 0.00 & 1.00 & 0.00 & 99.00 & 1.56 & 8.40\\
 &QS(0.5) & 0.00 & 0.00 & 0.10 & 0.10 & 99.80 & 7.50 & 3.54& & 0.00 & 0.00 & 0.10 & 0.10 & 99.80 & 7.49 & 3.54\\
  &NOT(HT)& 0.00 & 0.00 & 98.80 & 1.20 & 0.00 & 0.91 & 9.95& & 0.00 & 0.00 & 99.10 & 0.80 & 0.10 & 0.89 & 9.97\\
   &NOT(VAR)& 0.00 & 0.00 & 99.60 & 0.40 & 0.00 & 20.96 & 9.97& & 0.00 & 0.00 & 7.90 & 8.70 & 83.40 & 20.50 & 9.64\\
    &NWBS& 0.00 & 0.00 & 96.60 & 1.30 & 2.10 & 1.26 & 9.86& & 0.00 & 0.00 & 98.70 & 0.50 & 0.80 & 0.80 & 9.93\\
 \hline
  \multirow{5}{*}{Cauchy}&MQS(0.5)& 0.00 & 17.50 & 82.30 & 0.20 & 0.00 & 2.73 & 9.34 & \multirow{5}{*}{\shortstack{$\chi^2_3$\\$\sigma^2=0.04$}}& 0.00 & 0.00 & 99.70 & 0.30 & 0.00 & 3.53 & 9.34\\
  &SMUCE & 0.00 & 0.00 & 0.00 & 0.00 & 100.00 & 58.52 & 5.38& & 0.00 & 0.00 & 0.10 & 0.70 & 99.20 & 30.03 & 8.32\\
  &HSMUCE& 0.00 & 0.20 & 99.60 & 0.20 & 0.00 & 1.93 & 9.81& & 0.00 & 0.00 & 86.20 & 13.30 & 0.50 & 30.49 & 9.89\\
 &R-FPOP& 0.00 & 0.00 & 91.00 & 7.00 & 2.00 & 0.22 & 9.98& & 0.00 & 0.00 & 84.40 & 7.80 & 7.80 & 32.53 & 9.92\\
 &WBS& 0.00 & 0.00 & 0.00 & 0.00 & 100.00 & 81.71 & 6.94&& 0.00 & 0.00 & 26.20 & 0.30 & 73.50 & 30.30 & 9.42\\
 &QS(0.5)& 0.10 & 0.00 & 0.00 & 0.00 & 99.90 & 5.07 & 3.54& & 0.00 & 0.10 & 0.20 & 0.40 & 99.30 & 7.40 & 3.54\\
   &NOT (HT)& 56.20 & 1.80 & 41.60 & 0.20 & 0.20 & 56.25 & 6.62& & 0.00 & 0.00 & 99.40 & 0.30 & 0.30 & 30.29 & 9.96\\
   &NOT (VAR)& 0.00 & 0.00 & 0.00 & 0.10 & 99.90 & 121.06 & 8.95& & 0.00 & 0.00 & 61.20 & 15.20 & 23.60 & 29.26 & 9.89\\
    &NWBS& 0.00 & 0.10 & 99.10 & 0.40 & 0.40 & 0.19 & 9.97& & 0.00 & 0.00 & 98.20 & 0.90 & 0.90 & 1.03 & 9.93\\
\hline
\multirow{4}{*}{\shortstack{Normal\\$\sigma^2=0.04$}} &MQS(0.25)& 81.40 & 18.60 & 0.00 & 0.00 & 0.00 & 8.19 & 8.64& \multirow{4}{*}{\shortstack{t (3 d.f.)\\$\sigma^2=0.04$}}& 82.80 & 17.20 & 0.00 & 0.00 & 0.00 & 7.76 & 8.64\\
 &MQS(0.75)& 0.00 & 0.00 & 99.90 & 0.10 & 0.00 & 6.06 & 9.02& & 0.00 & 0.00 & 100.00 & 0.00 & 0.00 & 5.68 & 9.06\\
  &QS(0.25)& 0.30 & 0.10 & 1.80 & 3.60 & 94.20 & 8.75 & 3.55& & 0.00 & 0.10 & 1.00 & 2.30 & 96.60 & 7.58 & 3.55\\
 &QS(0.75)& 0.00 & 0.00 & 0.10 & 0.30 & 99.60 & 6.83 & 3.54& & 0.00 & 0.00 & 0.10 & 0.10 & 99.80 & 6.02 & 3.54\\
 \hline
 \multirow{4}{*}{Cauchy}&MQS(0.25)& 82.80 & 17.20 & 0.00 & 0.00 & 0.00 & 7.76 & 8.64& \multirow{4}{*}{\shortstack{$\chi^2_3$\\$\sigma^2=0.04$}} & 82.50 & 17.50 & 0.00 & 0.00 & 0.00 & 7.36 & 8.74\\
 &MQS(0.75)& 0.00 & 0.00 & 100.00 & 0.00 & 0.00 & 5.68 & 9.06& & 0.00 & 0.50 & 99.50 & 0.00 & 0.00 & 6.84 & 8.94\\
  &QS(0.25)& 0.00 & 0.10 & 1.00 & 2.30 & 96.60 & 7.58 & 3.55 & & 0.00 & 0.10 & 1.00 & 1.80 & 97.10 & 7.07 & 3.54\\
 &QS(0.75)& 0.00 & 0.00 & 0.10 & 0.10 & 99.80 & 6.02 & 3.54 & & 0.00 & 0.00 & 0.00 & 0.30 & 99.70 & 8.23 & 3.54\\
\end{tabular}
}
\caption{\label{tab:simpletable} \small Frequencies of estimated number of segments in $[\%]$, MIAE $(\times100)$, and V-measure $(\times10)$ for data as in Figure \ref{fig:bigex}. Here, the true number of segments equals 7. The proposed MQS estimator is compared with SMUCE \citep{frick2014}, HSMUCE  \citep{pein2017a}, WBS \citep{fryzlewicz2014}, R-FPOP \citep{fearnhead2017}, QS \citep{eilers2005}, NOT \citep{baranowski2019}, and NWBS \citep{padilla2019}. 
\smallskip
}
\end{table}


\begin{figure}
\includegraphics[width=\textwidth]{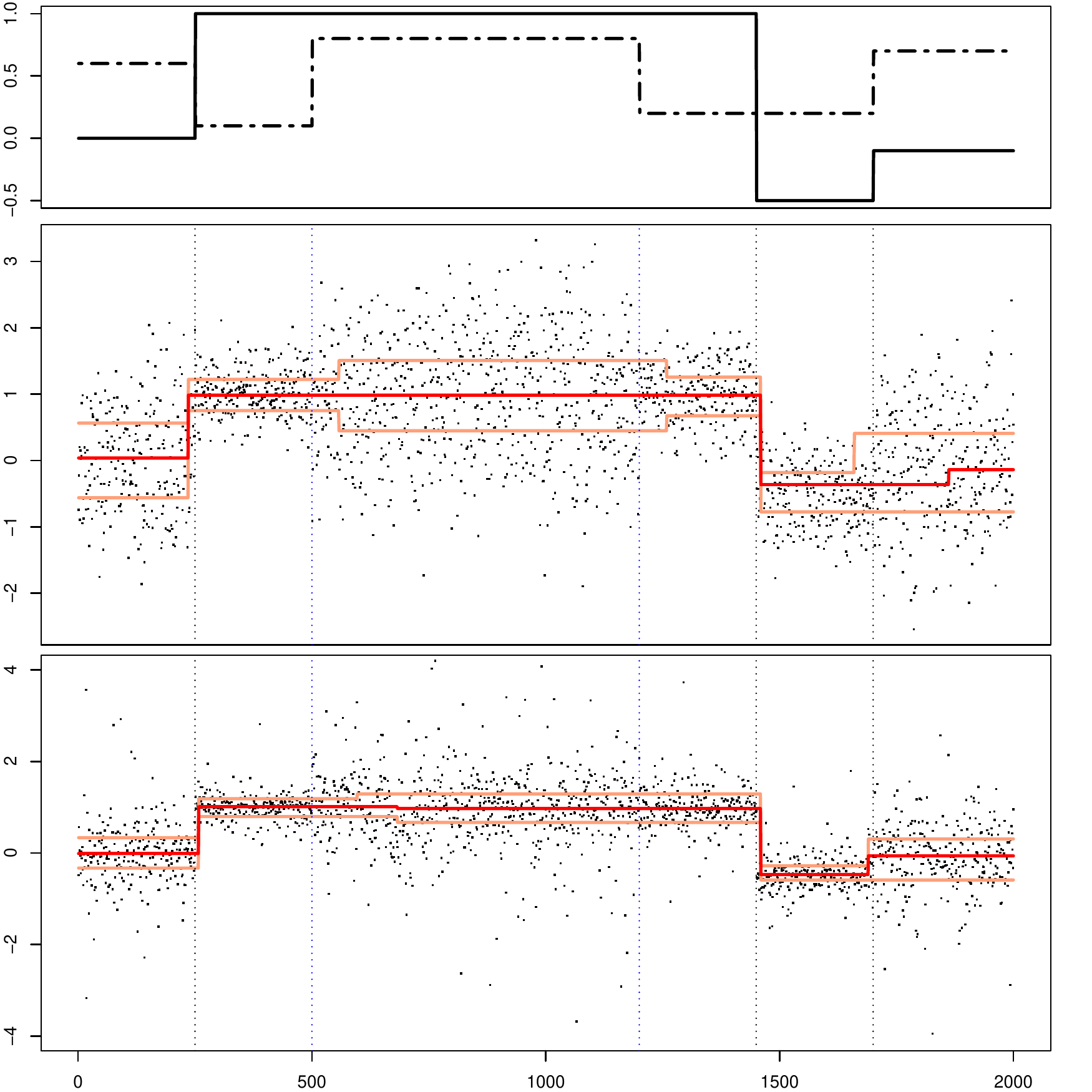}
\caption{ Top row: underlying mean (solid line) and variance (dashed line) functions. Subsequent rows: normally and $t$ distributed (from top to bottom) observations (black dots) with $n=2,000$ and mean and variance as in top row, together with the MQS box plot (red lines). The true location of segment changes is shown as a vertical dotted lines, the blue dotted lines correspond to changes only in the 0.25- and 0.75-quantiles.
}
\label{fig:varianceex}
\end{figure}

\begin{table}
\centering
\resizebox{\textwidth}{!}{%
\begin{tabular}{ c|c|c >{\columncolor[gray]{0.8}}c c >{\columncolor[gray]{0.8}}cc|c|c |c|c >{\columncolor[gray]{0.8}}c c >{\columncolor[gray]{0.8}}cc|c|c} 
 &Method&$\leq3$&4&5&6&$\geq7$&MIAE&V-m.&&$\leq3$&4&5&6&$\geq7$&MIAE&V-m.\\
 \hline
\multirow{6}{*}{Normal}&MQS(0.5)&16.08 & 83.11 & 0.81 & 0.00 & 0.00 & 6.13 & 8.67& \multirow{6}{*}{t (3 d.f.)} & 0.00 & 99.39 & 0.61 & 0.00 & 0.00 & 3.92 & 8.88\\
&SMUCE& 0.30 & 33.06 & 37.41 & 19.82 & 9.40 & 1.97 & 9.77& & 0.00 & 0.00 & 0.00 & 0.00 & 100.00 & 13.38 & 5.41\\
&HSMUCE& 24.57 & 73.61 & 1.82 & 0.00 & 0.00 & 1.96 & 9.97& & 13.09 & 85.99 & 0.92 & 0.00 & 0.00 & 2.19 & 9.95\\
&R-FPOP& 0.10 & 65.52 & 8.70 & 14.36 & 11.32 & 1.31 & 9.77& & 0.00 & 80.16 & 7.67 & 8.79 & 3.37 & 1.80 & 8.45 \\
&WBS& 0.61 & 90.60 & 1.62 & 6.17 & 1.01 & 2.02 & 9.77& & 0.00 & 0.00 & 0.00 & 0.72 & 99.28 & 2.72 & 7.75\\
&QS(0.5)& 0.00 & 0.00 & 0.10 & 0.51 & 99.39 & 9.93 & 2.55& & 0.00 & 0.00 & 0.00 & 0.51 & 99.49 & 6.49 & 2.55\\
   &NOT (HT)& 0.61 & 97.37 & 0.81 & 1.11 & 0.10 & 2.05 & 9.63& & 1.53 & 97.34 & 0.10 & 0.61 & 0.41 & 2.19 & 9.95\\
   &NOT (VAR)& 0.00 & 0.00 & 0.00 & 99.70 & 0.30 & 28.40 & 9.31& & 0.00 & 0.00 & 0.00 & 8.18 & 91.82 & 29.20 & 9.11\\
    &NWBS& 2.83 & 29.52 & 36.50 & 20.32 & 10.82 & 3.04 & 7.58& & 0.20 & 64.93 & 16.36 & 10.12 & 8.38 & 1.76 & 7.83\\
\hline
\multirow{4}{*}{Normal}&MQS(0.25)& 0.51 & 65.82 & 33.47 & 0.20 & 0.00 & 11.98 & 7.33& \multirow{4}{*}{t (3 d.f.)}& 2.04 & 57.98 & 35.79 & 4.19 & 0.00 & 12.78 & 6.81\\
&MQS(0.75)& 0.00 & 0.30 & 57.53 & 42.16 & 0.00 & 19.19 & 6.95& & 0.00 & 2.76 & 70.25 & 26.99 & 0.00 & 11.84 & 7.33\\
&QS(0.25)& 0.20 & 0.71 & 2.63 & 7.28 & 89.18 & 13.82 & 3.64& & 0.31 & 3.07 & 6.54 & 9.10 & 80.98 & 11.42 & 3.64\\
&QS(0.75)& 0.00 & 0.00 & 0.00 & 0.00 & 100.00 & 10.98 & 3.66& & 0.00 & 0.00 & 0.00 & 0.00 & 100.00 & 9.73 & 3.64\\
\end{tabular}
}
\caption{\label{tab:variance} Frequencies of estimated number of segments (in percentage), MIAE$(\times100)$ and V-measure$(\times10)$ for simultaneous changes in mean and variance, for data as in Figure \ref{fig:varianceex}. The true number of segments for the median is 4 and for the other quantiles 6.
\smallskip}
\end{table}

\begin{figure}[b!]
\centering
\includegraphics[width=0.95\textwidth]{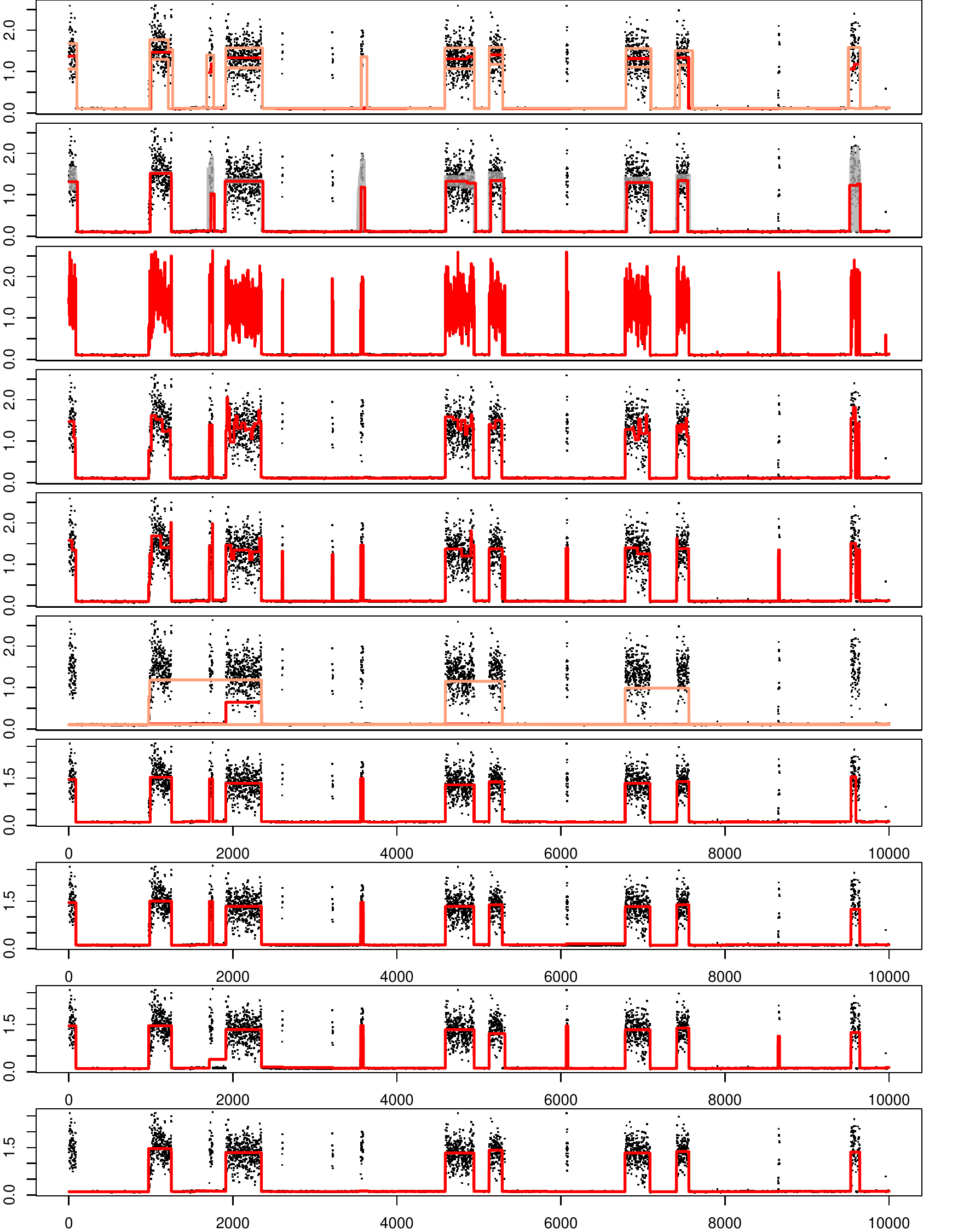}
  \caption{
Ion channel data (black dots) from a single channel of the bacterial porin PorB from the Steinam lab (Institute of Organic and Biomolecular Chemistry, University of G\"ottingen).
From top to bottom: MSB with $\alpha = 0.1$, MQSE for the median with confidence bands,
SMUCE \citep{frick2014}, R-FPOP \citep{fearnhead2017}, WBS \citep{fryzlewicz2014}, QS \citep{eilers2005} (together with 0.25- and 0.75-quantiles), HSMUCE \citep{pein2017a}, NOT(HT) and NOT (VAR) \citep{baranowski2019}, and NWBS \citep{padilla2019}.
}
\label{fig:ic}
\end{figure}
\end{appendices}
\bibliographystyle{chicago}  
\bibliography{MQR_Paper}
\end{document}